\documentclass[12pt,a4paper]{article}%
\usepackage{amsmath}
\usepackage{amsfonts}
\usepackage{amssymb}
\usepackage{graphicx}
\usepackage{tikz, tkz-euclide}%
\providecommand{\U}[1]{\protect\rule{.1in}{.1in}}
\newtheorem{theorem}{Theorem}

\newtheorem{lemma}[theorem]{Lemma}

\newtheorem{proposition}[theorem]{Proposition}

\newenvironment{proof}[1][Proof]{\noindent\textbf{#1.} }{\ \rule{0.5em}{0.5em}}
\begin{document}

\title{Uniform bound of the entanglement for the ground state of the quantum Ising
model with large transverse magnetic field}
\author{M. Campanino$^{\#}$, M. Gianfelice$%
%TCIMACRO{\U{b0}}%
%BeginExpansion
{{}^\circ}%
%EndExpansion
$\\$^{\#}$Dipartimento di Matematica\\Universit\`{a} degli Studi di Bologna\\P.zza di Porta San Donato, 5 \ I-40127\\massimo.campanino@unibo.it\\$^{%
%TCIMACRO{\U{b0}}%
%BeginExpansion
{{}^\circ}%
%EndExpansion
}$Dipartimento di Matematica e Informatica\\Universit\`{a} della Calabria\\Campus di Arcavacata\\Ponte P. Bucci - cubo 30B\\I-87036 Arcavacata di Rende\\gianfelice@mat.unical.it}
\maketitle

\begin{abstract}
We consider the ground state of the quantum Ising model with transverse field
$h$ in one dimension in a finite volume
\[
\Lambda{_{m}:=\{-m,-m+1,\ldots,m+L\}\ .}%
\]
For $h$ sufficiently large we prove a bound for the entanglement of the
interval $\Lambda_{0}:=\left\{  0,..,L\right\}  $ relative to its complement
$\Lambda_{m}\backslash\Lambda_{0}$ which is uniform in $m$ and $L$. The bound
is established by means of a suitable cluster expansion.

\end{abstract}
\tableofcontents

\bigskip

\begin{description}
\item[AMS\ subject classification:] {\small 60K35, 82B10, 82B31. }

\item[Keywords and phrases:] {\small Quantum Ising model, Entanglement,
Spin-flip processes, Gibbs random fields, Cluster expansion. }

\item[Acknowledgement:] {\small M. Gianfelice and M. Campanino are partially
supported by G.N.A.M.P.A. and by MIUR-PRIN 20155PAWZB-004 }\emph{Large Scale
Random Structures.}
\end{description}

\bigskip

\section{Introduction and results}

A characteristic feature distinguishing quantum systems from the classical
ones is that pure states of composite systems do not assign in general a
definite pure state to their subsystems. In particular, in the framework of
quantum statistical mechanics, the density matrix describing the state of a
subsystem, which goes under the name of \textit{reduced density matrix}, is
obtained taking the trace over the other system's components and in general
corresponds to a mixed state. This property is called entanglement and the von
Neumann entropy of the reduced density matrix of a subsystem can be considered
as a measure of the entanglement of the subsystem itself. One can therefore
address the problem of estimating the the von Neumann entropy (also known in
the literature as \textit{entanglement entropy}) of a subsystem w.r.t. the
ground state of the entire system. There are few rigorous results in this
direction, we refer the reader to \cite{GOS} and reference therein for a more
complete discussion on this topic.

The ground state of the quantum Ising model with a transverse magnetic field
can be represented as a classical Ising model with one added continuous
dimension \cite{DLP}. In turn this classical Ising model can be represented
via a suitable FK random cluster model \cite{FK}, \cite{CKP}, \cite{AKN}. This
last representation has been used for example in \cite{GOS} to study the
entanglement of the ground state in the supercritical regime.

In this paper we study the problem of entanglement for the supercritical
quantum Ising model with transverse magnetic field by using the representation
of \cite{DLP}. We consider a Gibbs random field in $\mathbb{Z}^{2}$ in which
the spins take values in the space of trajectories of a spin-flip process. For
this model a cluster expansion was developed in \cite{CG} and it was proved
that it satisfies the conditions for convergence (see \cite{KP}) when the
parameter $h$ corresponding to the strength of the transverse magnetic field
is sufficiently large.

We consider the ground state of the quantum Ising model with transverse field
$h$ in one dimension in a finite volume
\begin{equation}
\Lambda{_{m}:=\{-m,-m+1,\ldots,m+L\}\ .}%
\end{equation}
By using this cluster expansion we prove that for $h$ sufficiently large the
entanglement of the interval $\Lambda_{0}:=\left\{  0,..,L\right\}  $ relative
to its complement $\Lambda_{m}\backslash\Lambda_{0}$ is bounded by a constant
uniformly in $m$ and $L.$

It was proved in \cite{GOS} that for $h$ larger than some value corresponding
to a percolation threshold the entanglement is bounded by a constant times
$\log L.$

In section \ref{System} we recall the definition of the quantum Ising model
with transverse field on $\mathbb{Z}$.

In section \ref{Cluster} we recall and present in a more complete form the
result about the cluster expansion for the one-dimensional interacting
spin-flip process given in \cite{CG}. We remark that, although in this paper
all the computations are carried out for the one-dimensional model with
nearest-neighbour translation-invariant ferromagnetic couplings, the cluster
expansion presented in section \ref{Cluster} can be performed for the model
defined on $\mathbb{Z}^{d},d\geq1,$ with bounded, finite-range, pairwise interactions.

In section \ref{Entanglement} we recall the set up developed in \cite{GOS} in
order to estimate the entanglement entropy of the ground state of the system
and prove the key estimates which will lead us to the uniform bound of this quantity.

If $\mathcal{H}_{m}:=\mathcal{H}_{\Lambda_{m}}$ is the Hilbert space for the
quantum system defined on $\Lambda_{m},$ considering the representation of
$\mathcal{H}_{m}$ as $\mathcal{H}_{m,L}\otimes\mathcal{H}_{L},$ with
$\mathcal{H}_{L}:=\mathcal{H}_{\Lambda_{0}}$ and $\mathcal{H}_{m,L}%
:=\mathcal{H}_{\Lambda_{m}\backslash\Lambda_{0}},$ let $\rho_{m}^{L}$ be the
trace over $\mathcal{H}_{m,L}$ of the density operator associated to the
ground state of the system.

We will prove the following

\begin{theorem}
\label{t1}Consider a one-dimensional quantum Ising model in a transverse
magnetic field. There exists a positive value of the external magnetic field
$h^{\ast}$ such that, for any $h>h^{\ast},$ the \emph{entanglement entropy} of
the ground state $S\left(  \rho_{m}^{L}\right)  :=-tr_{\mathcal{H}_{L}}\left(
\rho_{m}^{L}\log\rho_{m}^{L}\right)  $ is bounded by a constant uniformly in
$m,L$ with $m\geq0,L\geq1.$
\end{theorem}

We stress that the cluster expansion can be carried out for one-dimensional
quantum Ising models with transverse field with bounded, finite-range,
translation-invariant, ferromagnetic interactions. Therefore our result can be
generalised in a straightforward way to this case.

\subsection{Notation}

Given a set $A\subset\mathbb{R}^{d},d\geq1,$ let us denote by $A^{c}$ its
complement. We also set $\mathcal{P}\left(  A\right)  $ to be the collection
of all subsets of $A,$ $\mathcal{P}_{n}\left(  A\right)  :=\{B\in
\mathcal{P}\left(  A\right)  :\left\vert B\right\vert =n\}$ and $\mathcal{P}%
_{f}\left(  A\right)  :=\bigcup_{n\geq1}\mathcal{P}_{n}\left(  A\right)  ,$
where $\left\vert B\right\vert $ is the cardinality of $B.$ Given $B\subset
A,\mathbf{1}_{B}$ denotes the indicator function of $B.$ Hence, if $B$ is a
discrete set, for any $b\in B,\mathbf{1}_{\left\{  b\right\}  }\left(
b^{\prime}\right)  =\left\{
\begin{array}
[c]{ll}%
1 & b=b^{\prime}\\
0 & b\neq b^{\prime}%
\end{array}
\right.  .$

A sequence of $\left\{  \Lambda_{n}\right\}  _{n\in\mathbb{N}}\in
\mathcal{P}_{f}\left(  \mathbb{Z}^{d}\right)  ,d\geq1,$ is called
\emph{cofinal} if for any $n\geq1,\Lambda_{n}\subset\Lambda_{n+1}$ and
$\left\{  \Lambda_{n}\right\}  \uparrow\mathbb{Z}^{d}.$

Let $\mathcal{B}\left(  \mathbb{R}^{d}\right)  $ be the Borel $\sigma$-algebra
of $\mathbb{R}^{d}$ and $\lambda^{d}$ be the Lebesgue measure on $\left(
\mathbb{R}^{d},\mathcal{B}\left(  \mathbb{R}^{d}\right)  \right)  .$ A
sequence of $\left\{  I_{n}\right\}  _{n\in\mathbb{N}}\in\mathcal{B}\left(
\mathbb{R}^{d}\right)  $ is called an \emph{exhaustion} of $\mathbb{R}^{d}$
if, for any $n\geq1,\lambda^{d}\left(  I_{n}\right)  <\infty,I_{n}\subset
I_{n+1}$ and $\left\{  I_{n}\right\}  \uparrow\mathbb{R}^{d}.$

For any $x\in\mathbb{R}^{d},$ we set $\left\vert x\right\vert :=\sum_{i=1}%
^{d}\left\vert x_{i}\right\vert $ and consequently, for any $A\subset
\mathbb{R}^{d},\mathrm{dist}\left(  x,A\right)  :=\inf_{y\in A}\left\vert
x-y\right\vert .$

Given a Hilbert space $\mathcal{H},$ let $\mathfrak{B}\left(  \mathcal{H}%
\right)  $ be the Banach space of bounded linear operators on $\mathcal{H}$
with norm $\left\Vert \mathbf{A}\right\Vert :=\sup_{\psi\in\mathcal{H}%
\ :\ \left\Vert \psi\right\Vert =1}\left\Vert \mathbf{A}\psi\right\Vert $ and
$\mathfrak{T}\left(  \mathcal{H}\right)  \subset\mathfrak{B}\left(
\mathcal{H}\right)  $ the collection of trace class operators on $\mathcal{H}$
which is a Banach space if endowed with the norm $tr_{\mathcal{H}}\left(
\left\vert \mathbf{A}\right\vert \right)  ,$ where $tr_{\mathcal{H}}\left(
\mathbf{A}\right)  $ denotes the trace of $\mathbf{A}\in\mathfrak{T}\left(
\mathcal{H}\right)  .$ Then, if $\mathcal{H}_{1},\mathcal{H}_{2}$ are Hilbert
spaces, we denote by $tr_{\mathcal{H}_{1}}:\mathfrak{T}\left(  \mathcal{H}%
_{1}\otimes\mathcal{H}_{2}\right)  \longmapsto\mathfrak{T}\left(
\mathcal{H}_{2}\right)  $ the partial trace w.r.t. $\mathcal{H}_{1}.$

\subsubsection{Graphs}

Let $G=\left(  V,E\right)  $ be a graph whose set of vertices and set of edges
are given respectively by a finite or enumerable set $V$ and $E\subset
\mathcal{P}_{2}\left(  V\right)  .\ G^{\prime}=\left(  V^{\prime},E^{\prime
}\right)  $ such that $V^{\prime}\subseteq V$ and $E^{\prime}\subseteq
\mathcal{P}_{2}\left(  V^{\prime}\right)  \cap E$ is said to be a subgraph of
$G$ and this property is denoted by $G^{\prime}\subseteq G.$ If $G^{\prime
}\subseteq G,$ we denote by $V\left(  G^{\prime}\right)  $ and $E\left(
G^{\prime}\right)  $ respectively the set of vertices and the collection of
the edges of $G^{\prime}.\ \left\vert V\left(  G^{\prime}\right)  \right\vert
$ is called the \emph{order} of $G^{\prime}$ while $\left\vert E\left(
G^{\prime}\right)  \right\vert $ is called its \emph{size}. Given $G_{1}%
,G_{2}\subseteq G,$ we denote by $G_{1}\cup G_{2}:=\left(  V\left(
G_{1}\right)  \cup V\left(  G_{2}\right)  ,E\left(  G_{1}\right)  \cup
E\left(  G_{2}\right)  \right)  \subset G$ the \emph{graph union }of $G_{1}$
and $G_{2}.$ Given $e\in E,$ we denote by $V_{e}:=\left\{  v\in V:\mathbf{1}%
_{e}\left(  v\right)  =1\right\}  $ the set of \emph{endpoints} of $e,$ hence
$e=\left\{  v,v^{\prime}\right\}  $ iff $\left\{  v,v^{\prime}\right\}
=V_{e}.$ Moreover, given $E^{\prime}\subseteq E,$ we denote by $V\left(
E^{\prime}\right)  :=\bigcup_{e\in E^{\prime}}V_{e}.$

A \emph{path} in $G$ is a subgraph $\gamma$ of $G$ such that there is a
bijection $\left\{  0,..,\left\vert E\left(  \gamma\right)  \right\vert
\right\}  \ni i\longmapsto v\left(  i\right)  :=x_{i}\in V\left(
\gamma\right)  $ with the property that any $e\in E\left(  \gamma\right)  $
can be represented as $\left\{  x_{i-1},x_{i}\right\}  $ for
$i=1,..,\left\vert E\left(  \gamma\right)  \right\vert .$ Two distinct
vertices $x,y$ of $G$ are said to be \emph{connected} if there exists a path
$\gamma\subseteq G$ such that $x_{0}=x,\ x_{\left\vert E\left(  \gamma\right)
\right\vert }=y.$ Therefore, if $\gamma$ is a path in $G,$ we will denote by
$\left\vert \gamma\right\vert $ its length $\left\vert E\left(  \gamma\right)
\right\vert $ and by $end\left(  \gamma\right)  :=\left\{  v\in V\left(
\gamma\right)  :\sum_{e\in E\left(  \gamma\right)  }\mathbf{1}_{e}\left(
v\right)  =1\right\}  $ the collection of its \emph{endpoints}. Hence, for any
$e\in E,$ the graph $\left(  V_{e},e\right)  \subset G$ is a path of length
$1$ and $end\left(  \left(  V_{e},e\right)  \right)  =V_{e}.$ A graph $G$ is
said to be \emph{connected} if any two distinct elements of $V\left(
G\right)  $ are connected. The maximal connected subgraphs of $G$ are called
\emph{components} of $G.$ Two connected subgraph $G_{1},G_{2}\subset G$ are
connected by a path $\gamma$ in $G$ if $G_{1}\cup G_{2}$ is not connected and
$G_{1}\cup G_{2}\cup\gamma$ is a connected subgraph of $G.$

We denote by $\mathbb{L}^{d}$ the graph $\left(  \mathbb{Z}^{d},\mathbb{E}%
^{d}\right)  $ with $\mathbb{E}^{d}:=\left\{  \left\{  x,y\right\}
\in\mathcal{P}_{2}\left(  \mathbb{Z}^{d}\right)  :\left\vert x-y\right\vert
=1\right\}  .$ If $\Lambda\subset\mathbb{Z}^{d},$ we also set $\partial
\Lambda:=\left\{  y\in\Lambda^{c}:\mathrm{dist}\left(  y,\Lambda\right)
=1\right\}  $ and $\mathbb{L}_{\Lambda}^{d}:=\left(  \Lambda\cup
\partial\Lambda,\mathbb{E}_{\Lambda}^{d}\right)  ,$ where $\mathbb{E}%
_{\Lambda}^{d}:=\left\{  e\in\mathbb{E}^{d}:V_{e}\subset\left(  \Lambda
\cup\partial\Lambda\right)  \right\}  .$

\subsection{The model}

We consider the Hilbert space $\mathcal{H}:=l^{2}\left(  \left\{
-1,1\right\}  ,\mathbb{C}\right)  $ which is isomorphic to $\mathbb{C}^{2}.$
The algebra $\mathcal{U}:=M\left(  2,{\mathbb{C}}\right)  $ of bounded linear
operators on $\mathcal{H}$ is then generated by the Pauli matrices
$\sigma^{\left(  i\right)  },i=1,2,3$ and by the identity $I.$ In particular,
unless differently specified, in the following we will always consider the
representation of $\mathcal{U}$ with respect to which $\sigma^{\left(
3\right)  }$ is diagonal, i.e.
\begin{equation}
\sigma^{\left(  3\right)  }=\left(
\begin{array}
[c]{cc}%
1 & 0\\
0 & -1
\end{array}
\right)
\end{equation}
and
\begin{equation}
\sigma^{\left(  1\right)  }=\left(
\begin{array}
[c]{cc}%
0 & 1\\
1 & 0
\end{array}
\right)  \ .
\end{equation}

Let $\Lambda$ be a finite connected subset of $\mathbb{Z}$ and set
$\mathcal{H}_{\Lambda}:=\bigotimes_{x\in\Lambda}\mathcal{H}_{x}$ where, for
any $x\in\Lambda,\mathcal{H}_{x}$ is a copy of $\mathcal{H}$ at $x.$ The
finite volume Hamiltonian of the ferromagnetic quantum Ising model with
transverse field is the linear operator on $\mathcal{H}_{\Lambda}$%
\begin{equation}
\mathbf{H}_{\Lambda}\left(  J,h\right)  :=-\frac{1}{2}J\sum_{x,y\in
\Lambda\ :\ \left\{  x,y\right\}  \in\mathbb{E}}\sigma_{x}^{\left(  3\right)
}\sigma_{y}^{\left(  3\right)  }-h\sum_{x\in\Lambda}\sigma_{x}^{\left(
1\right)  }\ , \label{HtfI}%
\end{equation}
with $h>0$ and $J\geq0$ for any $x,y\in\Lambda.$

Given $\Lambda\subset\subset\mathbb{Z},$ it can be proven \cite{CKP} that
$\mathbf{H}_{\Lambda}\left(  J,h\right)  $ generates a positivity improving
semigroup which by the Perron-Frobenius theorem has a unique ground state
$\Psi_{\Lambda}\in\mathcal{H}_{\Lambda}.$ The same argument applies to the
operator $\mathbf{L}_{\Lambda}\left(  h\right)  :\mathcal{H}_{\Lambda
}\circlearrowleft$ such that
\begin{equation}
\mathbf{L}_{\Lambda}\left(  h\right)  :=h\sum_{x\in\Lambda}\left(  \sigma
_{x}^{\left(  1\right)  }-1\right)  \ ,
\end{equation}
whose ground state $\Psi_{\Lambda}^{0}\in\mathcal{H}_{\Lambda}$ is such that
$\left\langle \Psi_{\Lambda},\Psi_{\Lambda}^{0}\right\rangle >0$ and, for any
element $\mathbf{A}$ of the Abelian subalgebra $\mathfrak{A}_{\Lambda},$
generated by $\left\{  \sigma_{x}^{\left(  3\right)  },x\in\Lambda\right\}  ,$
of the algebra of linear operators on $\mathcal{H}_{\Lambda}$%
\begin{equation}
\left\langle \Psi_{\Lambda},\mathbf{A}\Psi_{\Lambda}\right\rangle =\lim
_{\beta\rightarrow\infty}\frac{\left\langle \Psi_{\Lambda}^{0},e^{-\frac
{\beta}{2}\mathbf{H}_{\Lambda}\left(  J,h\right)  }\mathbf{A}e^{-\frac{\beta
}{2}\mathbf{H}_{\Lambda}\left(  J,h\right)  }\Psi_{\Lambda}^{0}\right\rangle
}{\left\langle \Psi_{\Lambda}^{0},e^{-\beta\mathbf{H}_{\Lambda}\left(
J,h\right)  }\Psi_{\Lambda}^{0}\right\rangle }=\lim_{\beta\rightarrow\infty
}\frac{tr_{\mathcal{H}_{\Lambda}}\left(  e^{-\beta\mathbf{H}_{\Lambda}\left(
J,h\right)  }\mathbf{A}\right)  }{tr_{\mathcal{H}_{\Lambda}}\left(
e^{-\beta\mathbf{H}_{\Lambda}\left(  J,h\right)  }\right)  }\ .\label{gs}%
\end{equation}

\subsubsection{Spin-flip process description of the system}

In \cite{DLP} (Section 2.5) it has been shown that in the chosen
representation for the Pauli matrices, for any $h>0,$ the linear operator
$\mathbf{L}\left(  h\right)  :=h\left(  \sigma^{\left(  1\right)  }-1\right)
$ on $\mathcal{H}$ can be interpreted as the generator of a continuous time
Markov process with state space $\left\{  -1,1\right\}  ,$ the so called
\emph{spin-flip process}, with rate $h.$ Namely, for any function $f$ on
$\left\{  -1,1\right\}  $%
\begin{equation}
L\left(  h\right)  f\left(  \xi\right)  :=h\left(  f\left(  -\xi\right)
-f\left(  \xi\right)  \right)  \;,\;\xi\in\left\{  -1,1\right\}  \ .
\end{equation}
Hence, given a Poisson point process $\left(  N_{h}\left(  t\right)
,t\in\mathbb{R}\right)  $ with intensity $h,$ we can consider the random
process which, with a little abuse of notation, we denote by
\begin{equation}
\mathbb{R}\ni t\longmapsto\sigma\left(  t\right)  :=\left(  -1\right)
^{N_{h}\left(  t\right)  }\in\left\{  -1,1\right\}  \ ,
\end{equation}
that is the stationary measure $\mu$ defined by the semigroup generated by
$L\left(  h\right)  $ on the measurable space $\left(  \mathcal{D}%
,\mathcal{F}\right)  ,$ where $\mathcal{D}$ is the Skorokhod space
$\mathbb{D}\left(  \mathbb{R},\left\{  -1,1\right\}  \right)  $ of piecewise
$\left\{  -1,1\right\}  $-valued rcll (c\`{a}dl\`{a}g) constant functions on
$\mathbb{R}$ and $\mathcal{F}$ is the $\sigma$-algebra generated by the open
sets in the associated Skorokhod topology.

Consequently, for any interval $I\subset\mathbb{R},$ let $\mu_{I}$ be the
restriction of $\mu$ to the measurable space $\left(  \mathcal{D}%
_{I},\mathcal{F}_{I}\right)  ,$ where
\begin{equation}
\mathcal{D}_{I}:=\left\{  \sigma\in\mathbb{D}\left(  I,\left\{  -1,1\right\}
\right)  :\sigma=\sigma^{\prime}\upharpoonleft_{I},\ \sigma^{\prime}%
\in\mathcal{D}\right\}
\end{equation}
and $\mathcal{F}_{I}$ is the $\sigma$-algebra generated by the open sets in
the associated Skorokhod topology. Moreover, we denote by $\mu_{I}^{p}$ the
probability distribution corresponding to periodic b.c.'s, that is conditional
to
\begin{align}
\mathcal{D}_{I}^{p}  &  :=\left\{  \sigma\in\mathcal{D}_{I}:\sigma\left(
-\frac{\beta}{2}\right)  =\sigma\left(  \frac{\beta}{2}\right)  \right\} \\
&  =\left\{  \sigma\in\mathcal{D}_{I}:N_{h}\left(  \frac{\beta}{2}\right)
-N_{h}\left(  -\frac{\beta}{2}\right)  =2k\ ,\ k\in\mathbb{Z}^{+}\right\}
\ .\nonumber
\end{align}

Let $\mathfrak{D}_{I}:=\mathcal{D}_{I}^{\mathbb{Z}}$ the configuration space
of the random field $\mathbb{Z}\ni x\longmapsto\sigma_{x}\in\mathcal{D}%
_{I},\mathfrak{F}_{I}:=\mathcal{F}_{I}^{\otimes\mathbb{Z}}$ the $\sigma
$-algebra generated by the cylinder events of $\mathfrak{D}_{I},$ and $\nu
_{I}$ the product measure $\mu_{I}^{\otimes\mathbb{Z}}.$

\paragraph{The finite volume distribution\label{System}}

Given $\beta>0,$ let us set $I:=\left[  -\frac{\beta}{2},\frac{\beta}%
{2}\right]  .$ For any finite subset $\Lambda$ of $\mathbb{Z}$ we denote by
$\sigma_{\Lambda}$ the restriction of the configuration $\sigma\in
\mathfrak{D}_{I}$ to $\mathcal{D}_{I}^{\Lambda}$ and set $\sigma_{\Lambda
}\left(  t\right)  :=\left\{  \sigma_{x}\left(  t\right)  \right\}
_{x\in\Lambda},\mathcal{F}_{I}^{\Lambda}:=\mathcal{F}_{I}^{\otimes\Lambda}$
and $\mu_{I}^{\Lambda}:=\mu_{I}^{\otimes\Lambda}.$ We introduce the
conditional Gibbs measure $\nu_{I,\Lambda}^{\eta;\xi^{+},\xi^{-}}$ on $\left(
\mathcal{D}_{I}^{\Lambda},\mathcal{F}_{I}^{\Lambda}\right)  $ with density
w.r.t. $\mu_{I}^{\Lambda}$ given by
\begin{align}
&  Z_{\Lambda,I}^{-1}\left(  \eta;\xi^{+},\xi^{-}\right)  \exp\left[
J\sum_{x,y\in\Lambda\ :\ \left\{  x,y\right\}  \in\mathbb{E}}\int
_{-\frac{\beta}{2}}^{\frac{\beta}{2}}\text{d}t\sigma_{x}\left(  t\right)
\sigma_{y}\left(  t\right)  +J\sum_{x\in\Lambda}\sum_{y\in\partial\Lambda}%
\int_{-\frac{\beta}{2}}^{\frac{\beta}{2}}\text{d}t\sigma_{x}\left(  t\right)
\eta_{y}\left(  t\right)  \right]  \times\label{Gibbs1}\\
&  \times\prod\limits_{x\in\Lambda}\mathbf{1}_{\left\{  \xi_{x}^{-}\right\}
}\left(  \sigma_{x}\left(  -\frac{\beta}{2}\right)  \right)  \mathbf{1}%
_{\left\{  \xi_{x}^{+}\right\}  }\left(  \sigma_{x}\left(  \frac{\beta}%
{2}\right)  \right)  \ ,\nonumber
\end{align}
where $\eta\in\mathcal{D}_{I}^{\Lambda^{c}},\xi^{+},\xi^{-}\in\Omega_{\Lambda
}:=\left\{  -1,1\right\}  ^{\Lambda}$ and $Z_{\Lambda,I}\left(  \eta;\xi
^{+},\xi^{-}\right)  $ is the normalizing constant.

In \cite{DLP} it has been shown that the expected value of an observable
$\mathbf{F}\in\mathfrak{A}_{\Lambda}$ in the equilibrium (KMS) state of the
ferromagnetic quantum Ising model with transverse field at inverse temperature
$\beta>0$ can be represented as the expected value w.r.t. the Gibbs
distribution (\ref{Gibbs1}) with periodic b.c.'s at $t=\pm\frac{\beta}{2}$ of
the function $F$ on $\Omega_{\Lambda}$ corresponding to the spectral
representation of $\mathbf{F}$ computed at $\sigma_{\Lambda}\left(  0\right)
\in\Omega_{\Lambda}.$ Namely
\begin{equation}
\frac{tr_{\mathcal{H}_{\Lambda}}\left(  e^{-\beta\mathbf{H}_{\Lambda}\left(
h\right)  }\mathbf{F}\right)  }{tr_{\mathcal{H}_{\Lambda}}\left(
e^{-\beta\mathbf{H}_{\Lambda}\left(  h\right)  }\right)  }=\frac
{tr_{\mathcal{H}_{\Lambda}}\left(  e^{-\frac{\beta}{2}\mathbf{H}_{\Lambda
}\left(  h\right)  }\mathbf{F}e^{-\frac{\beta}{2}\mathbf{H}_{\Lambda}\left(
h\right)  }\right)  }{tr_{\mathcal{H}_{\Lambda}}\left(  e^{-\beta
\mathbf{H}_{\Lambda}\left(  h\right)  }\right)  }=\nu_{I,\Lambda}^{\eta
}\left(  F\left(  \sigma_{\Lambda}\left(  0\right)  \right)  \right)  \ ,
\end{equation}
where
\begin{equation}
\frac{\text{d}\nu_{I,\Lambda}^{\eta}}{\text{d}\mu_{I,\Lambda}^{p}}%
:=Z_{\Lambda,I}^{-1}\left(  \eta\right)  \exp\left[  J\sum_{x,y\in
\Lambda\ :\ \left\{  x,y\right\}  \in\mathbb{E}}\int_{-\frac{\beta}{2}}%
^{\frac{\beta}{2}}\text{d}t\sigma_{x}\left(  t\right)  \sigma_{y}\left(
t\right)  +J\sum_{x\in\Lambda}\sum_{y\in\partial\Lambda}\int_{-\frac{\beta}%
{2}}^{\frac{\beta}{2}}\text{d}t\sigma_{x}\left(  t\right)  \eta_{y}\left(
t\right)  \right]
\end{equation}
is the density of the conditional Gibbs measure w.r.t. $\mu_{I,\Lambda}%
^{p}:=\left(  \mu_{I}^{p}\right)  ^{\otimes\Lambda}.$ Clearly the b.c.
$\eta\in\mathcal{D}_{I}^{\Lambda^{c}}$ can be thought of as a time varying
local external field in the direction of the spin field (see also \cite{CKP}
and \cite{KL} for a general discussion).

Correlation inequalities imply that the expected value of local observables of
the form $\prod\limits_{x\in\Lambda}\sigma_{x}^{\left(  z\right)  }%
,\Lambda\subset\subset\mathbb{Z},$ in the ground state of the ferromagnetic
quantum Ising model with transverse field can be computed from $\nu
_{I_{n},\Lambda_{m}}^{\eta}\left[  \prod\limits_{x\in\Lambda}\sigma_{x}\left(
0\right)  \right]  $ by taking first the limit through an exhaustion $\left\{
I_{n}\right\}  _{n\in\mathbb{N}}$ of $\mathbb{R}$ (i.e. $\beta\rightarrow
\infty$) and then the limit $\left\{  \Lambda_{m}\right\}  _{m\in\mathbb{N}%
}\uparrow\mathbb{Z}$ (see also \cite{CKP} Section 2). For sufficiently large
values of the external field $h,$ these limits can be shown to be independent
of the b.c.'s by using the cluster expansion carried out in the next section.

Therefore, in the following, we will consider fixed b.c.'s at $\left\{
\left(  x,t\right)  \in\Lambda\times I:t=\pm\frac{\beta}{2}\right\}  ,$ free
b.c.'s at $\left\{  \left(  x,t\right)  \in\Lambda^{c}\times I\right\}  $ and
assume that $\beta$ is a multiple of $\delta.$

\section{Cluster expansion\label{Cluster}}

We perform a cluster expansion on the model and verify that, when $h$ is
sufficiently large, we can ensure that, for a suitable choice of the
parameters, the condition of Koteck\'{y} and Preiss \cite{KP} are satisfied
and the cluster expansion is therefore convergent.

We stress that the following argument applies to a more general setup in which
the model is defined on $\mathbb{Z}^{d},$ with $d\geq1,$ and the coupling
between any pair of spins are bounded. We also remark that the requirement for
the two-body interactions to be ferromagnetic is needed in order to guarantee,
by means of correlation inequalities, the existence of the ground state, while
translation-invariance and finite-rangeness are sufficient conditions for the
existence of thermodynamics.

Given $\delta$ to be fixed later, we partition the trajectory of any spin-flip
process $\left(  \sigma_{x}\left(  t\right)  ,t\in I\right)  ,x\in\Lambda,$
into blocks of size $\delta$ (fig. 1). We will call the last coordinate of the
vector in $\mathbb{R}^{2}$ corresponding to a point in $\mathbb{Z}\times
\delta\mathbb{Z}$ the \emph{vertical component}. Then, denoting an element $x$
of $\mathbb{Z}\times\delta\mathbb{Z}$ by $x=\left(  x_{1},\delta x_{2}\right)
,$ we denote by $\mathbb{L}_{\delta}^{2}$ the graph whose set of vertices is
$\mathbb{Z}\times\delta\mathbb{Z}$ and whose set of edges is $\mathbb{E}%
_{\delta}^{2}:=\left\{  \left\{  x,y\right\}  \in\mathcal{P}_{2}\left(
\mathbb{Z}\times\delta\mathbb{Z}\right)  :\left\vert x_{1}-y_{1}\right\vert
+\left\vert x_{2}-y_{2}\right\vert =1\right\}  .$

Let us set $\mathbb{V}:=\left\{  \left\{  x,y\right\}  \in\mathbb{E}_{\delta
}^{2}:x_{1}=y_{1}\right\}  $ the set of vertical edges in $\mathbb{L}_{\delta
}^{2}$ and by $\mathbb{O}:=\mathbb{E}_{\delta}^{2}\backslash\mathbb{V}.$
Denoting by $\Delta:=\Lambda\times\left(  \delta\mathbb{Z}\cap I\right)  $ we
define $\mathbb{O}_{\Delta}:=\left\{  e\in\mathbb{O}:V_{e}\subset\left(
\Delta\backslash\overline{\partial}\Delta\right)  \right\}  $ and
$\mathbb{V}_{\Delta}:=\left\{  e\in\mathbb{V}:V_{e}\subset\Delta\right\}  .$
Moreover, we define
\begin{equation}
\partial^{\pm}\Delta:=\left\{  \left(  x_{1},\delta x_{2}\right)  \in
\Delta:x_{1}\in\Lambda,\delta x_{2}=\pm\frac{\beta}{2}\right\}  \label{b+-}%
\end{equation}
and set $\overline{\partial}\Delta:=\partial^{+}\Delta\cup\partial^{-}\Delta$
and%
\begin{equation}
\partial\Delta:=\overline{\partial}\Delta\cup\left\{  x\in\mathbb{Z}%
\times\delta\mathbb{Z}:x_{1}\in\partial\Lambda,x_{2}\in\delta\mathbb{Z}\cap
I\right\}  \ . \label{bDelta}%
\end{equation}

Then, denoting by $\Omega_{D}:=\left\{  -1,1\right\}  ^{D},$ for any
$D\subset\mathbb{Z}\times\delta\mathbb{Z},$ assuming b.c. $\xi=\left(  \xi
^{+},\xi^{-}\right)  \in\Omega_{\overline{\partial}\Delta}:=\Omega
_{\partial^{+}\Delta}\times\Omega_{\partial^{-}\Delta}$ at $\overline
{\partial}\Delta,$ with $\xi^{+},\xi^{-}$ appearing in (\ref{Gibbs1}), and
free b.c.'s at $\partial\Delta\backslash\overline{\partial}\Delta,$ we have
\begin{align}
Z_{\Delta}\left(  \xi\right)   &  :=Z_{\Lambda,I}\left(  \xi^{+},\xi
^{-}\right)  =\int%
%TCIMACRO{\dbigotimes \limits_{z_{1}\in\Lambda}}%
%BeginExpansion
{\displaystyle\bigotimes\limits_{z_{1}\in\Lambda}}
%EndExpansion
\mu_{I}\left(  \text{d}\sigma_{z_{1}}\right)  e^{\sum_{x,y\in\Delta
\ :\ \left\{  x,y\right\}  \in\mathbb{O}}W\left(  \sigma_{x},\sigma
_{y}\right)  }\times\\
&  \times\prod\limits_{z_{1}\in\Lambda}\mathbf{1}_{\left\{  \xi_{z_{1}}%
^{-}\right\}  }\left(  \sigma_{z_{1}}\left(  -\frac{\beta}{2}\right)  \right)
\mathbf{1}_{\left\{  \xi_{z_{1}}^{+}\right\}  }\left(  \sigma_{z_{1}}\left(
\frac{\beta}{2}\right)  \right)  \ ,\nonumber
\end{align}
where, for any $x,y\in\Delta\backslash\overline{\partial}\Delta,$%
\begin{equation}
W\left(  \sigma_{x},\sigma_{y}\right)  =J\int_{0}^{\delta}dt\sigma_{x_{1}%
}\left(  \delta x_{2}+t\right)  \sigma_{y_{1}}\left(  \delta y_{2}+t\right)
\ .
\end{equation}
Setting
\begin{equation}
e^{W\left(  \sigma_{x},\sigma_{y}\right)  }=1+\left[  e^{W\left(  \sigma
_{x},\sigma_{y}\right)  }-1\right]  \ ,
\end{equation}
$Z_{\Delta}\left(  \xi\right)  $ can be rewritten as
\begin{align}
Z_{\Delta}\left(  \xi\right)   &  =\sum_{\ell\in\mathcal{P}\left(
\mathbb{O}_{\Delta}\right)  }\int%
%TCIMACRO{\dbigotimes \limits_{z_{1}\in\Lambda}}%
%BeginExpansion
{\displaystyle\bigotimes\limits_{z_{1}\in\Lambda}}
%EndExpansion
\mu_{I}\left(  \text{d}\sigma_{z_{1}}\right)  \prod\limits_{e\in\ell}\left[
e^{\mathbf{1}_{e}\left(  \left\{  x,y\right\}  \right)  W\left(  \sigma
_{x},\sigma_{y}\right)  }-1\right]  \times\\
&  \times\mathbf{1}_{\left\{  \xi^{-}\right\}  }\left(  \sigma\left(
-\frac{\beta}{2}\right)  \right)  \mathbf{1}_{\left\{  \xi^{+}\right\}
}\left(  \sigma\left(  \frac{\beta}{2}\right)  \right) \nonumber\\
&  =\sum_{\ell\in\mathcal{P}\left(  \mathbb{O}_{\Delta}\right)  }\int%
%TCIMACRO{\dbigotimes \limits_{z_{1}\in\Lambda}}%
%BeginExpansion
{\displaystyle\bigotimes\limits_{z_{1}\in\Lambda}}
%EndExpansion
\mu_{I}\left(  \text{d}\sigma_{z_{1}}\right)  \prod\limits_{e\in\ell
}\mathbf{1}_{e}\left(  \left\{  x,y\right\}  \right)  \left(  e^{W\left(
\sigma_{x},\sigma_{y}\right)  }-1\right)  \times\nonumber\\
&  \times\mathbf{1}_{\left\{  \xi^{-}\right\}  }\left(  \sigma\left(
-\frac{\beta}{2}\right)  \right)  \mathbf{1}_{\left\{  \xi^{+}\right\}
}\left(  \sigma\left(  \frac{\beta}{2}\right)  \right)  \ .\nonumber
\end{align}
Given $\ell\in\mathcal{P}\left(  \mathbb{O}_{\Delta}\right)  ,$ for any
$x_{1}\in\Lambda$ we can integrate over the trajectories of the stationary
process $\left(  \sigma_{x_{1}}\left(  t\right)  ,t\in I\right)  $ keeping
fixed its values at $\delta x_{2}$ if $\left(  x_{1},\delta x_{2}\right)
\in\ell.$ This integral can be computed explicitely. Indeed, setting $V\left(
\ell\right)  :=\left(
%TCIMACRO{\tbigcup \limits_{e\in\ell}}%
%BeginExpansion
{\textstyle\bigcup\limits_{e\in\ell}}
%EndExpansion
V_{e}\right)  $ and denoting by
\begin{equation}
V^{\prime}\left(  \ell\right)  :=%
%TCIMACRO{\dbigcup \limits_{e\in\ell}}%
%BeginExpansion
{\displaystyle\bigcup\limits_{e\in\ell}}
%EndExpansion
\left\{  z\in\Delta\cup\overline{\partial}\Delta:z_{1}=x_{1},z_{2}%
=x_{2}+1,\left(  x_{1},\delta x_{2}\right)  \in V_{e}\right\}  \ ,
\end{equation}
we have
\begin{gather}
\mu_{I,\Lambda}^{\xi}\left[  \prod\limits_{e\in\ell}\mathbf{1}_{e}\left(
\left\{  x,y\right\}  \right)  \left(  e^{W\left(  \sigma_{x},\sigma
_{y}\right)  }-1\right)  \right]  =\\
\mu_{I,\Lambda}^{\xi}\left[  \mu_{I,\Lambda}^{\xi}\left[  \prod\limits_{e\in
\ell}\mathbf{1}_{e}\left(  \left\{  x,y\right\}  \right)  \left(  e^{W\left(
\sigma_{x},\sigma_{y}\right)  }-1\right)  |\left\{  \sigma_{x_{1}}\left(
\delta x_{2}\right)  \right\}  _{\left(  x_{1},\delta x_{2}\right)  \in
V\left(  \ell\right)  \cup V^{\prime}\left(  \ell\right)  }\right]  \right]
\ ,\nonumber
\end{gather}
where we have set $\mu_{I,\Lambda}^{\xi}:=\mu_{I,\Lambda}\left[  \cdot
|\sigma\left(  \pm\frac{\beta}{2}\right)  =\xi^{\pm}\right]  .$

If $\left(  x_{1},\delta x_{2}\right)  ,\left(  x_{1},\delta y_{2}\right)  \in
V\left(  \ell\right)  $ such that $y_{2}\geq x_{2}+2,$ and there is no other
$\left(  x_{1},\delta z_{2}\right)  \in V\left(  \ell\right)  $ such that
$x_{2}+2\leq z_{2}\leq y_{2}-1,$ we can integrate over the trajectories of
$\left(  \sigma_{x_{1}}\left(  t\right)  ,t\in I\right)  $ with given values
at $t=\delta x_{2}+\delta,\delta y_{2}.$ Let
\begin{align}
x_{2}^{\left(  1\right)  }  &  :=\min\left\{  z_{2}\in\mathbb{Z}:\left(
x_{1},\delta z_{2}\right)  \in V\left(  \ell\right)  \right\}  \ ,\\
x_{2}^{\left(  i+1\right)  }  &  :=\min\left\{  z_{2}\in\mathbb{Z}:\left(
x_{1},\delta z_{2}\right)  \in V\left(  \ell\right)  \backslash%
%TCIMACRO{\dbigcup \limits_{j=1}^{i}}%
%BeginExpansion
{\displaystyle\bigcup\limits_{j=1}^{i}}
%EndExpansion
\left(  x_{1},\delta x_{2}^{\left(  j\right)  }\right)  \right\}
\;,\;i\geq1\ .
\end{align}
Then, $T_{\ell}\left(  x_{1}\right)  :=\left\{  z_{2}\in\mathbb{Z}:\left(
x_{1},\delta z_{2}\right)  \in V\left(  \ell\right)  \right\}  $ can be
represented as the ordered set $T_{\ell}\left(  x_{1}\right)  =\left\{
x_{2}^{\left(  1\right)  },..,x_{2}^{\left(  \left\vert T_{\ell}\left(
x_{i}\right)  \right\vert \right)  }\right\}  .$ For any $i=1,..,\left\vert
T_{\ell}\left(  x_{1}\right)  \right\vert ,$ we denote by
\begin{equation}
y_{2}^{\left(  i\right)  }:=\left\{  z_{2}\in\mathbb{Z}:\left(  x_{1},\delta
z_{2}\right)  \in\Delta\backslash V\left(  \ell\right)  ,z_{2}=x_{2}^{\left(
i\right)  }+1\right\}
\end{equation}
and set
\begin{align}
\overline{T}_{\ell}\left(  x_{1}\right)   &  :=\left\{  x_{2}^{\left(
1\right)  },..,x_{2}^{\left(  \left\vert V_{\ell}\left(  x_{i}\right)
\right\vert \right)  },x_{2}^{\left(  \left\vert V_{\ell}\left(  x_{i}\right)
\right\vert +1\right)  }:=\frac{\beta}{2\delta}\right\}  \ ,\\
\Gamma_{\ell}\left(  x_{1}\right)   &  :=\left\{  -\frac{\beta}{2\delta
}=:y_{2}^{\left(  0\right)  },y_{2}^{\left(  1\right)  },..,y_{2}^{\left(
\left\vert V_{\ell}\left(  x_{1}\right)  \right\vert \right)  }\right\}  \ .
\end{align}
Hence, denoting by
\begin{equation}
V_{\ell}\left(  x_{1}\right)  :=%
%TCIMACRO{\dbigcup \limits_{x_{2}\in T_{\ell}\left(  x_{1}\right)  }}%
%BeginExpansion
{\displaystyle\bigcup\limits_{x_{2}\in T_{\ell}\left(  x_{1}\right)  }}
%EndExpansion
\left\{  x\in\Delta:x=\left(  x_{1},\delta x_{2}\right)  \right\}  \ ,
\end{equation}
we get
\begin{gather}
\int\mu_{I}^{\xi_{x_{1}}}\left(  \text{d}\sigma_{x_{1}}|\left\{  \sigma
_{x_{1}}\left(  \delta x_{2}\right)  \right\}  _{x_{2}\in T_{x_{1}}\left(
\ell\right)  \cup\Gamma_{x_{1}}\left(  \ell\right)  }\right)  \prod
\limits_{e\in\ell\ :\ V_{e}\cap V_{\ell}\left(  x_{1}\right)  \neq\varnothing
}\mathbf{1}_{e}\left(  \left\{  x,y\right\}  \right)  \left(  e^{W\left(
\sigma_{x},\sigma_{y}\right)  }-1\right)  =\\
\int%
%TCIMACRO{\dbigotimes \limits_{i=1}^{\left\vert V_{\ell}\left(  x_{1}\right)
%\right\vert }}%
%BeginExpansion
{\displaystyle\bigotimes\limits_{i=1}^{\left\vert V_{\ell}\left(
x_{1}\right)  \right\vert }}
%EndExpansion
\mu_{I}^{\xi_{x_{1}}}\left(  \text{d}\sigma_{x_{1}}|\sigma_{x_{1}}\left(
\delta x_{2}^{\left(  i\right)  }\right)  ,\sigma_{x_{1}}\left(  \delta
y_{2}^{\left(  i\right)  }\right)  \right)  \prod\limits_{e\in\ell
\ :\ V_{e}\cap V_{\ell}\left(  x_{1}\right)  \neq\varnothing}\mathbf{1}%
_{e}\left(  \left\{  x,y\right\}  \right)  \left(  e^{W\left(  \sigma
_{x},\sigma_{y}\right)  }-1\right)  \times\nonumber\\
\times\prod\limits_{i=0}^{\left\vert T_{\ell}\left(  x_{1}\right)  \right\vert
}\frac{1+\sigma_{x_{1}}\left(  \delta y_{2}^{\left(  i\right)  }\right)
\sigma_{x_{1}}\left(  \delta x_{2}^{\left(  i+1\right)  }\right)
e^{-2h\delta\left(  x_{2}^{\left(  i+1\right)  }-y_{2}^{\left(  i\right)
}\right)  }}{2}\ ,\nonumber
\end{gather}
where we have used that, given $x_{1}\in\Lambda,$ for any $t,s\in I$ with
$t>s,\eta,\eta^{\prime}\in\left\{  -1,1\right\}  ,$%
\begin{equation}
\mu_{I}\left[  \mathbf{1}_{\left\{  \eta^{\prime}\right\}  }\left(  \sigma
_{x}\left(  t\right)  \right)  |\sigma_{x}\left(  s\right)  =\eta\right]
=\frac{1+\eta^{\prime}\eta e^{-2h\left(  t-s\right)  }}{2}=\left\{
\begin{array}
[c]{ll}%
\frac{1+e^{-2h\left(  y_{2}-x_{2}\right)  }}{2} & \text{if }\eta^{\prime}%
=\eta\\
\frac{1-e^{-2h\left(  t-s\right)  }}{2} & \text{if }\eta^{\prime}=-\eta
\end{array}
\right.  \ .
\end{equation}
Therefore, setting
\begin{equation}
\Lambda\left(  \ell\right)  :=\left\{  x_{1}\in\Lambda:\left\vert V_{\ell
}\left(  x_{1}\right)  \right\vert \geq1\right\}  \ ,
\end{equation}
since $\mu_{\Lambda,I}^{\xi}=%
%TCIMACRO{\dbigotimes \limits_{x_{1}\in\Lambda\left(  \ell\right)  }}%
%BeginExpansion
{\displaystyle\bigotimes\limits_{x_{1}\in\Lambda\left(  \ell\right)  }}
%EndExpansion
\mu_{I}^{\xi_{x_{1}}},$ we obtain
\begin{align}
Z_{\Delta}\left(  \xi\right)   &  =\sum_{\ell\in\mathcal{P}\left(
\mathbb{O}_{\Delta}\right)  }\int%
%TCIMACRO{\dbigotimes \limits_{x_{1}\in\Lambda\left(  \ell\right)  }}%
%BeginExpansion
{\displaystyle\bigotimes\limits_{x_{1}\in\Lambda\left(  \ell\right)  }}
%EndExpansion
\mu_{I}^{\xi_{x_{1}}}\left(  \text{d}\sigma_{x_{1}}|\left\{  \sigma_{x_{1}%
}\left(  \delta x_{2}\right)  \right\}  _{x_{2}\in T_{x_{1}}\left(
\ell\right)  \cup\Gamma_{x_{1}}\left(  \ell\right)  }\right)  \times\\
&  \times\prod\limits_{x_{1}\in\Lambda\backslash\Lambda\left(  \ell\right)
}\frac{1+\xi_{x_{1}}^{+}\xi_{x_{1}}^{-}e^{-2h\beta}}{2}\prod\limits_{e\in\ell
}\mathbf{1}_{e}\left(  \left\{  x,y\right\}  \right)  \left(  e^{W\left(
\sigma_{x},\sigma_{y}\right)  }-1\right)  \times\nonumber\\
&  \times\prod\limits_{i=0}^{\left\vert T_{\ell}\left(  x_{1}\right)
\right\vert }\frac{1+\sigma_{x_{1}}\left(  \delta y_{2}^{\left(  i\right)
}\right)  \sigma_{x_{1}}\left(  \delta x_{2}^{\left(  i+1\right)  }\right)
e^{-2h\delta\left(  x_{2}^{\left(  i+1\right)  }-y_{2}^{\left(  i\right)
}\right)  }}{2}\ .\nonumber
\end{align}

It can be useful to represent $Z_{\Delta}\left(  \xi\right)  $ as the
partition function of a classical spin system. Indeed, we can consider a
classical spin system on $\mathbb{Z}\times\delta\mathbb{Z}$ by associating to
any lattice point $\left(  x_{1},\delta x_{2}\right)  \in\mathbb{Z}%
\times\delta\mathbb{Z}$ a random element, which we will still call
\emph{spin}, taking values in the space $\mathcal{D}_{\delta}$ of piecewise
$\left\{  -1,1\right\}  $-valued functions on $\left[  0,\delta\right]  $
endowed with the Skorokhod topology, namely
\begin{equation}
\mathcal{D}_{\delta}:=\left\{  \sigma\in\mathbb{D}\left(  \left[
0,\delta\right]  ,\left\{  -1,1\right\}  \right)  \right\}  \ .
\end{equation}
Setting $\mathcal{S}:=\mathcal{D}_{\delta}^{\mathbb{Z}},$ we denote by
$\mathbf{S}$ the injection of $\mathcal{D}$ in $\mathcal{S}$ such that
\begin{equation}
\mathcal{D}\ni\sigma\longmapsto\mathbf{S}\left(  \sigma\right)  :=\left\{
\sigma^{\left(  k\right)  }\right\}  _{k\in\mathbb{Z}}\in\mathcal{S}\ ,
\end{equation}
where $\forall k\in\mathbb{Z},\sigma^{\left(  k\right)  }$ denotes the element
of $\mathcal{D}_{\delta}$ representing the function $\left[  0,\delta\right]
\ni t\longmapsto\sigma^{\left(  k\right)  }\left(  t\right)  :=\sigma\left(
k\delta+t\right)  \in\left\{  -1,1\right\}  .$ Equipping $\mathcal{S}$ with
the product topology, the push-forward of $\mu$ w.r.t. $\mathbf{S}$ on
$\left(  \mathcal{S},\mathfrak{S}\right)  $ with $\mathfrak{S}$ the product
$\sigma$-algebra can be written as
\begin{equation}
\mu\circ\mathbf{S}^{-1}\left(  \text{d}\left\{  \sigma^{\left(  k\right)
}\right\}  _{k\in\mathbb{Z}}\right)  =2^{\frac{\beta}{\delta}-1}%
\bigotimes\limits_{k\in\mathbb{Z}}\mu^{\delta}\left(  \text{d}\sigma^{\left(
k\right)  }\right)  \prod\limits_{k\in\mathbb{Z}}\delta_{\sigma^{\left(
k\right)  }\left(  \delta\right)  ,\sigma^{\left(  k+1\right)  }\left(
0\right)  }\ , \label{v1}%
\end{equation}
with
\begin{equation}
\mu^{\delta}\left(  \text{d}\sigma^{\left(  k\right)  }\right)  :=\mu_{\left[
\delta k,\delta\left(  k+1\right)  \right]  }\left(  \text{d}\sigma\right)
\;,\;k\in\mathbb{Z}\ .
\end{equation}

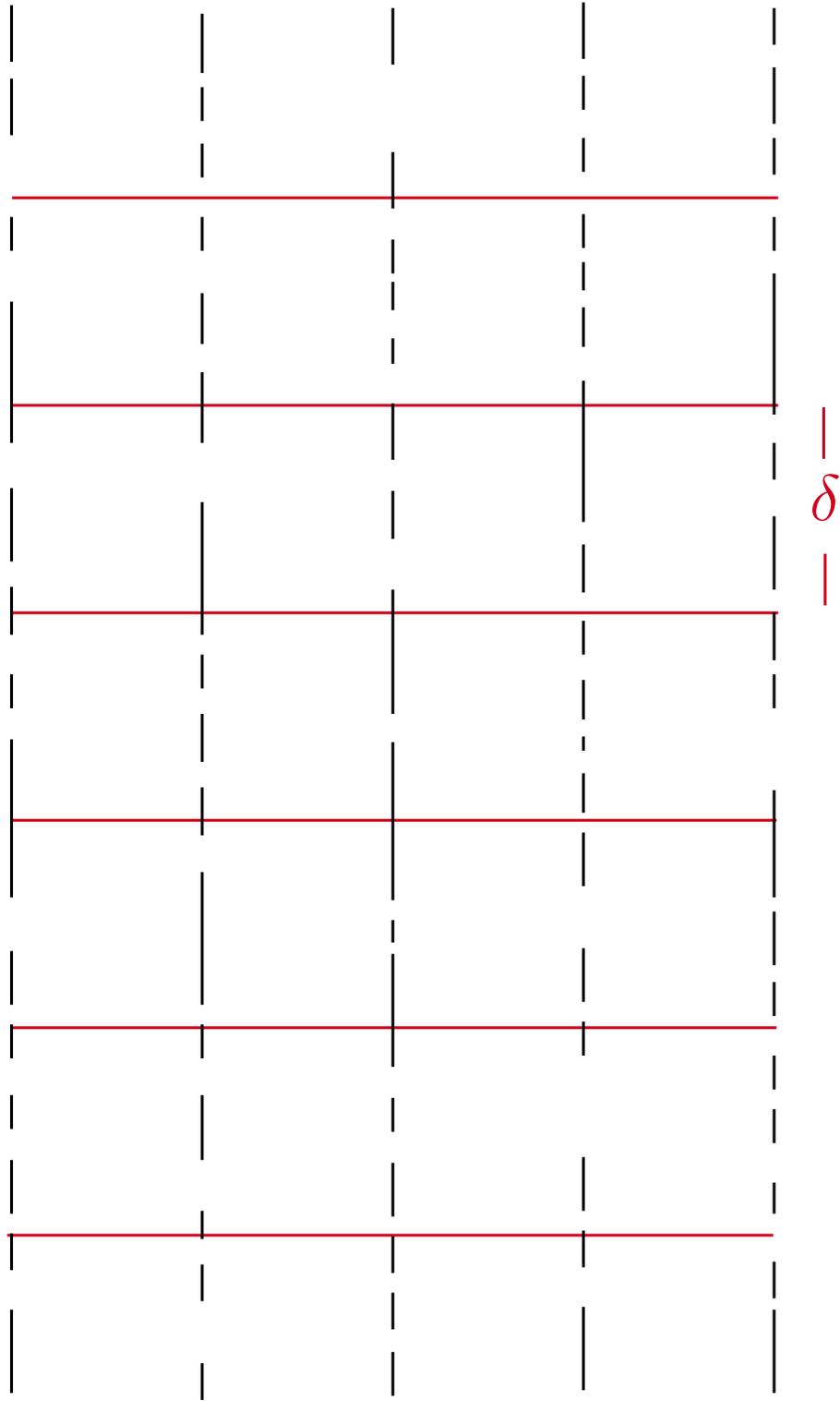
\begin{figure}[ptbh]
\centering
\resizebox{0.75\textwidth}{!}{
\tikzset{every picture/.style={line width=0.75pt}}
\begin{tikzpicture}[x=0.75pt,y=0.75pt,yscale=-1,xscale=1]
\draw [color={rgb, 255:red, 208; green, 2; blue, 27 }  ,draw opacity=1 ]   (85.26,74.23) -- (354.5,74.23) ;
\draw [color={rgb, 255:red, 208; green, 2; blue, 27 }  ,draw opacity=1 ]   (85.26,147.71) -- (354.5,147.71) ;
\draw [color={rgb, 255:red, 208; green, 2; blue, 27 }  ,draw opacity=1 ]   (85.26,221.19) -- (354.5,221.19) ;
\draw [color={rgb, 255:red, 208; green, 2; blue, 27 }  ,draw opacity=1 ]   (84.67,294.68) -- (353.91,294.68) ;
\draw [color={rgb, 255:red, 208; green, 2; blue, 27 }  ,draw opacity=1 ]   (84.67,368.16) -- (353.91,368.16) ;
\draw [color={rgb, 255:red, 208; green, 2; blue, 27 }  ,draw opacity=1 ]   (83.5,441.64) -- (352.74,441.64) ;
\draw [color={rgb, 255:red, 208; green, 2; blue, 27 }  ,draw opacity=1 ]   (371,200.31) -- (371,218.57) ;
\draw [color={rgb, 255:red, 208; green, 2; blue, 27 }  ,draw opacity=1 ]   (370.5,148.4) -- (370.5,166.65) ;
\draw    (152.06,452) -- (152.06,465) ;
\draw    (152.06,433) -- (152.06,443) ;
\draw    (152.06,392) -- (152.06,415) ;
\draw    (152.06,367) -- (152.06,379) ;
\draw    (152.06,313) -- (152.06,360) ;
\draw    (152.06,257) -- (152.06,274) ;
\draw    (152.06,236) -- (152.06,248) ;
\draw    (152.06,182) -- (152.06,229) ;
\draw    (152.06,136) -- (152.06,161) ;
\draw    (152.06,93) -- (152.06,81) ;
\draw    (152.06,67) -- (152.06,55) ;
\draw    (152.06,47) -- (152.06,35) ;
\draw    (85.06,468) -- (85.06,497.53) ;
\draw    (85.06,441) -- (85.06,454) ;
\draw    (85.06,415) -- (85.06,434) ;
\draw    (85.06,392) -- (85.06,404) ;
\draw    (85.06,367) -- (85.06,379) ;
\draw    (85.06,341) -- (85.06,360) ;
\draw    (85.06,266) -- (85.06,322) ;
\draw    (85.06,243) -- (85.06,255) ;
\draw    (85.06,212) -- (85.06,229) ;
\draw    (85.06,177) -- (85.06,203) ;
\draw    (85.06,111) -- (85.06,161) ;
\draw    (85.06,93) -- (85.06,81) ;
\draw    (85.06,52) -- (85.06,32) ;
\draw    (85.06,26) -- (85.06,6) ;
\draw    (219.06,483) -- (219.06,498.53) ;
\draw    (219.06,442) -- (219.06,455) ;
\draw    (219.06,416) -- (219.06,435) ;
\draw    (219.06,393) -- (219.06,405) ;
\draw    (219.06,342) -- (219.06,382) ;
\draw    (219.06,267) -- (219.06,323) ;
\draw    (219.06,213) -- (219.06,257) ;
\draw    (219.06,178) -- (219.06,195) ;
\draw    (219.06,147) -- (219.06,167) ;
\draw    (219.06,101) -- (219.06,89) ;
\draw    (219.06,78) -- (219.06,58) ;
\draw    (219.06,27) -- (219.06,7) ;
\draw    (286.06,467) -- (286.06,496.53) ;
\draw    (286.06,440) -- (286.06,453) ;
\draw    (286.06,414) -- (286.06,433) ;
\draw    (286.06,366) -- (286.06,378) ;
\draw    (286.06,340) -- (286.06,359) ;
\draw    (286.06,265) -- (286.06,270) ;
\draw    (286.06,224) -- (286.06,236) ;
\draw    (286.06,197) -- (286.06,214) ;
\draw    (286.06,97) -- (286.06,107) ;
\draw    (286.06,92) -- (286.06,80) ;
\draw    (286.06,43) -- (286.06,31) ;
\draw    (286.06,25) -- (286.06,5) ;
\draw    (353.06,468) -- (353.06,497.53) ;
\draw    (353.06,451) -- (353.06,464) ;
\draw    (353.06,423) -- (353.06,434) ;
\draw    (353.06,397) -- (353.06,409) ;
\draw    (353.06,352) -- (353.06,364) ;
\draw    (353.06,327) -- (353.06,346) ;
\draw    (353.06,284) -- (353.06,322) ;
\draw    (353.06,243) -- (353.06,255) ;
\draw    (353.06,221) -- (353.06,238) ;
\draw    (353.06,187) -- (353.06,213) ;
\draw    (353.06,101) -- (353.06,151) ;
\draw    (353.06,93) -- (353.06,81) ;
\draw    (353.06,66) -- (353.06,53) ;
\draw    (353.06,48) -- (353.06,28) ;
\draw    (152.06,30) -- (152.06,9) ;
\draw    (152.06,108) -- (152.06,126) ;
\draw    (152.06,283) -- (152.06,300) ;
\draw    (152.06,487) -- (152.06,500) ;
\draw    (219.06,104) -- (219.06,114) ;
\draw    (219.06,124) -- (219.06,133) ;
\draw    (219.06,462) -- (219.06,475) ;
\draw    (219.06,330) -- (219.06,338) ;
\draw    (286.06,65) -- (286.06,53) ;
\draw    (286.06,139) -- (286.06,189) ;
\draw    (286.06,113) -- (286.06,127) ;
\draw    (286.06,245) -- (286.06,259) ;
\draw    (286.06,278) -- (286.06,292) ;
\draw    (286.06,299) -- (286.06,318) ;
\draw    (353.06,20) -- (353.06,7) ;
\draw    (353.06,161) -- (353.06,174) ;
\draw    (353.06,378) -- (353.06,390) ;
\draw (371,180.35) node  [font=\Large,color={rgb, 255:red, 208; green, 2; blue, 27 }  ,opacity=1 ]  {$\delta $};
\end{tikzpicture}}
%If not, use
%\vspace{5cm}       % Give the correct figure height in cm
%Give a unique label
\caption{The construction of the \emph{spins} in $\mathcal{D}_{\delta}.$}%
\end{figure}

For any $x=\left(  x_{1},\delta x_{2}\right)  \in\Delta,$ with a little abuse
of notation we denote by $\sigma_{x}$ the element of $\mathcal{D}_{\delta}$
representing the function $\left[  0,\delta\right]  \ni t\longmapsto
\sigma_{x_{1}}\left(  \delta x_{2}+t\right)  \in\left\{  -1,1\right\}  .$
Hence, we denote by $\mathcal{S}_{\Delta}:=\mathcal{D}_{\delta}^{\Delta}$ and
by $\mathfrak{S}_{\Delta}:=\left\{  A\cap\mathcal{S}_{\Delta}:A\in
\mathfrak{S}\right\}  .$ In particular, we can represent the Gibbs probability
measure $\nu_{I,\Lambda}$ on $\left(  \mathcal{D}_{I}^{\Lambda},\mathcal{F}%
_{I}^{\Lambda}\right)  $ specified by (\ref{Gibbs1}), with fixed b.c. $\xi
\in\Omega_{\overline{\partial}\Delta}$ at $\overline{\partial}\Delta$ and free
b.c.'s at $\partial\Delta\backslash\overline{\partial}\Delta,$ by the Gibbs
probability measure $\nu_{\delta}^{\xi}\left(  \text{d}\sigma_{\Delta}\right)
$ on $\left(  \mathcal{S}_{\Delta},\mathfrak{S}_{\Delta}\right)  $ specified
by the density
\begin{equation}
Z_{\delta}^{-1}\left(  \xi\right)  \exp\left[  \sum_{x,y\in\Delta}\left(
W_{1}\left(  \sigma_{x},\sigma_{y}\right)  +\mathbf{1}_{\mathbb{O}_{\Delta}%
}\left(  \left\{  x,y\right\}  \right)  W\left(  \sigma_{x},\sigma_{y}\right)
\right)  \right]  \mathbf{1}_{\left\{  \xi^{-}\right\}  }\left(
\sigma_{\partial^{-}\Delta}\right)  \mathbf{1}_{\left\{  \xi^{+}\right\}
}\left(  \sigma_{\partial^{+}\Delta}\right)  \label{Gibbs2}%
\end{equation}
w.r.t. the reference measure $\mu^{\delta}\left(  \text{d}\sigma_{\Delta
}\right)  :=\bigotimes\limits_{x\in\Delta}\mu^{\delta}\left(  \text{d}%
\sigma_{x}\right)  ,$ associated to the interaction $W_{1}+W,$ where, in view
of the fact that, by (\ref{v1}), for any $x_{1},x_{2}\in\mathbb{Z},$ the spins
$\sigma_{\left(  x_{1},\delta x_{2}\right)  },\sigma_{\left(  x_{1},\delta
x_{2}+\delta\right)  }\in\mathcal{D}_{\delta}$ must satisfy the compatibility
condition $\sigma_{\left(  x_{1},\delta x_{2}\right)  }\left(  \delta\right)
=\sigma_{\left(  x_{1},\delta x_{2}+\delta\right)  }\left(  0\right)  ,$

\begin{enumerate}
\item $W_{1}\left(  \sigma_{x},\sigma_{y}\right)  =0$ if $\left\{
x,y\right\}  \in\mathbb{V}$ and if $x_{2}<y_{2},\sigma_{x}\left(
\delta\right)  =\sigma_{y}\left(  0\right)  $ or if $y_{2}<x_{2},\sigma
_{y}\left(  \delta\right)  =\sigma_{x}\left(  0\right)  ;$

\item $W_{1}\left(  \sigma_{x},\sigma_{y}\right)  =-\infty$ if $\left\{
x,y\right\}  \in\mathbb{V}$ and if $x_{2}<y_{2},\sigma_{x}\left(
\delta\right)  \neq\sigma_{y}\left(  0\right)  $ or if $y_{2}<x_{2},\sigma
_{y}\left(  \delta\right)  \neq\sigma_{x}\left(  0\right)  .$
\end{enumerate}

Then, by the definition of the potential $W_{1},$%
\begin{equation}
e^{W_{1}\left(  \sigma_{x},\sigma_{y}\right)  }=\left[  \delta_{x_{1},y_{1}%
}\left(  \delta_{y_{2},x_{2}+1}\delta_{\sigma_{x}\left(  \delta\right)
,\sigma_{y}\left(  0\right)  }+\delta_{x_{2},y_{2}+1}\delta_{\sigma_{y}\left(
\delta\right)  ,\sigma_{x}\left(  0\right)  }\right)  +\left(  1-\delta
_{x_{1},y_{1}}\right)  \right]  \ .
\end{equation}
Hence,
\begin{align}
Z_{\Delta}\left(  \xi\right)   &  =\sum_{\ell\in\mathcal{P}\left(
\mathbb{O}_{\Delta}\right)  }\prod\limits_{x_{1}\in\Lambda\left(  \ell\right)
}2^{\left\vert T_{\ell}\left(  x_{1}\right)  \right\vert }\int%
%TCIMACRO{\dbigotimes \limits_{x_{2}\in T_{\ell}\left(  x_{1}\right)  }}%
%BeginExpansion
{\displaystyle\bigotimes\limits_{x_{2}\in T_{\ell}\left(  x_{1}\right)  }}
%EndExpansion
\mu^{\delta}\left(  \text{d}\sigma_{\left(  x_{1},\delta x_{2}\right)
}\right)
%TCIMACRO{\dprod \limits_{x_{2}\in T_{\ell}\left(  x_{1}\right)  }}%
%BeginExpansion
{\displaystyle\prod\limits_{x_{2}\in T_{\ell}\left(  x_{1}\right)  }}
%EndExpansion
e^{W_{1}\left(  \sigma_{\left(  x_{1},\delta x_{2}\right)  },\sigma_{\left(
x_{1},\delta x_{2}+\delta\right)  }\right)  }\times\\
&  \times\prod\limits_{e\in\ell}\mathbf{1}_{e}\left(  \left\{  x,y\right\}
\right)  \left(  e^{W\left(  \sigma_{x},\sigma_{y}\right)  }-1\right)
\prod\limits_{\left(  x,y\right)  \in\partial^{+}\Delta\times\partial
^{-}\Delta\ :\ x_{1}=y_{1}\ ,\ x_{1}\in\Lambda\backslash\Lambda\left(
\ell\right)  }\frac{1+\xi_{x}\xi_{y}e^{-2h\beta}}{2}\times\nonumber\\
&  \times\prod\limits_{i=0}^{\left\vert T_{\ell}\left(  x_{1}\right)
\right\vert }\frac{1+\sigma_{\left(  x_{1},\delta y_{2}^{\left(  i\right)
}\right)  }\sigma_{\left(  x_{1},\delta x_{2}^{\left(  i+1\right)  }\right)
}e^{-2h\delta\left(  x_{2}^{\left(  i+1\right)  }-y_{2}^{\left(  i\right)
}\right)  }}{2}\mathbf{1}_{\left\{  \xi_{x_{1}}^{-}\right\}  }\left(
\sigma_{\left(  x_{1},\delta y_{2}^{\left(  0\right)  }\right)  }\right)
\mathbf{1}_{\left\{  \xi_{x_{1}}^{+}\right\}  }\left(  \sigma_{\left(
x_{1},\delta y_{2}^{\left(  \left\vert T_{\ell}\left(  x_{1}\right)
\right\vert +1\right)  }\right)  }\right)  \ .\nonumber
\end{align}
Let us denote by $\Pi\left(  \ell\right)  $ the set of paths in $\left(
\Delta,\mathbb{V}_{\Delta}\right)  $ connecting any couple of points $\left(
y,x\right)  \in V\left(  \ell\right)  \cup\partial^{-}\Delta\times V\left(
\ell\right)  \cup\partial^{+}\Delta$ such that:

\begin{itemize}
\item if $y=\left(  x_{1},y_{2}\right)  $ with $x_{1}\in\Lambda\left(
\ell\right)  ,x=\left(  x_{1},x_{2}\left(  y\right)  \right)  $ with
$x_{2}\left(  y\right)  :=\min\left\{  z_{2}\in\overline{T}_{\ell}\left(
x_{1}\right)  :z_{2}\geq y_{2}+2\right\}  ;$

\item if $y=\left(  x_{1},-\frac{\beta}{2}\right)  $ with $x_{1}\in
\Lambda\backslash\Lambda\left(  \ell\right)  ,x=\left(  x_{1},\frac{\beta}%
{2}\right)  .$
\end{itemize}

Hence, we can write
\begin{align}
Z_{\Delta}\left(  \xi\right)   &  =\sum_{\ell\in\mathcal{P}\left(
\mathbb{O}_{\Delta}\right)  }2^{\left\vert V\left(  \ell\right)  \right\vert
}\int%
%TCIMACRO{\dbigotimes \limits_{x\in V\left(  \ell\right)  }}%
%BeginExpansion
{\displaystyle\bigotimes\limits_{x\in V\left(  \ell\right)  }}
%EndExpansion
\mu^{\delta}\left(  \text{d}\sigma_{x}\right)
%TCIMACRO{\dprod \limits_{\left\{  x,y\right\}  \in\mathbb{V}_{\Delta}\ :\ x\in
%V\left(  \ell\right)  ,y\in V^{\prime}\left(  \ell\right)  }}%
%BeginExpansion
{\displaystyle\prod\limits_{\left\{  x,y\right\}  \in\mathbb{V}_{\Delta
}\ :\ x\in V\left(  \ell\right)  ,y\in V^{\prime}\left(  \ell\right)  }}
%EndExpansion
e^{W_{1}\left(  \sigma_{x},\sigma_{y}\right)  }\times\label{Zd2}\\
&  \times\prod\limits_{e\in\ell}\mathbf{1}_{e}\left(  \left\{  x,y\right\}
\right)  \left(  e^{W\left(  \sigma_{x},\sigma_{y}\right)  }-1\right)
\times\nonumber\\
&  \times\prod\limits_{\gamma\in\Pi\left(  \ell\right)  }\mathbf{1}%
_{end\left(  \gamma\right)  }\left(  \left\{  x,y\right\}  \right)
\frac{1+\sigma_{x}\sigma_{y}e^{-2h\delta\left\vert x-y\right\vert }}{2}%
\times\nonumber\\
&  \times\left[  \left(  1-\mathbf{1}_{\partial^{+}\Delta}\left(  x\right)
\right)  +\mathbf{1}_{\partial^{+}\Delta}\left(  x\right)  \mathbf{1}_{\xi
_{x}^{+}}\left(  \sigma_{x}\right)  \right]  \left[  \left(  1-\mathbf{1}%
_{\partial^{-}\Delta}\left(  y\right)  \right)  +\mathbf{1}_{\partial
^{-}\Delta}\left(  y\right)  \mathbf{1}_{\xi_{y}^{-}}\left(  \sigma
_{y}\right)  \right]  \ .\nonumber
\end{align}

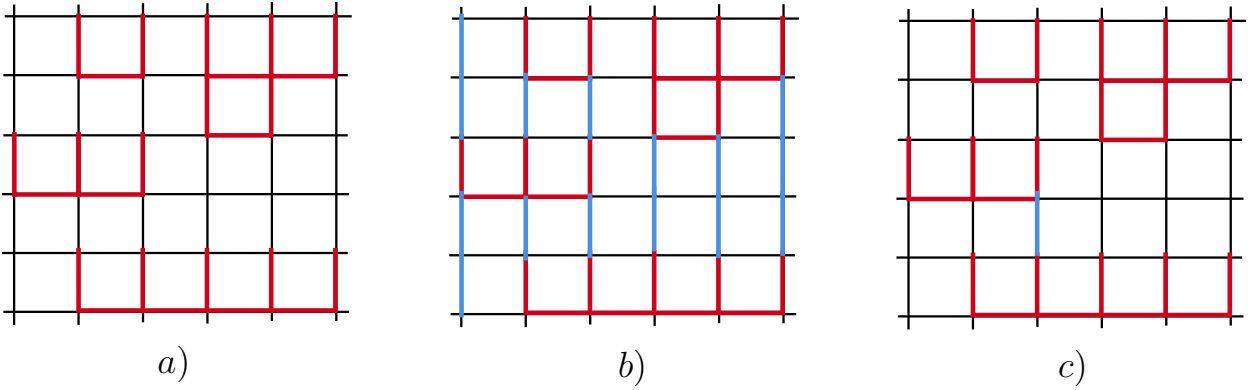
\begin{figure}[ptbh]
\centering
\resizebox{1.0\textwidth}{!}{
\tikzset{every picture/.style={line width=0.75pt}}
\begin{tikzpicture}[x=0.75pt,y=0.75pt,yscale=-1,xscale=1]
\draw    (45,128.51) -- (196.33,128.51) ;
\draw    (45,154.38) -- (195.74,154.38) ;
\draw    (50.27,45.15) -- (50.27,186) ;
\draw    (45.59,102.64) -- (195.74,102.64) ;
\draw    (45.59,76.19) -- (195.74,76.19) ;
\draw    (46.17,50.32) -- (197.5,50.32) ;
\draw    (78.37,45.15) -- (78.37,186) ;
\draw    (106.47,44.57) -- (106.47,184.85) ;
\draw    (45.59,180.25) -- (196.33,180.25) ;
\draw    (134.57,44.57) -- (134.57,185.43) ;
\draw    (162.67,44) -- (162.67,184.28) ;
\draw    (190.77,44) -- (190.77,184.28) ;
\draw [color={rgb, 255:red, 208; green, 2; blue, 27 }  ,draw opacity=1 ][line width=1.5]    (78.37,49.15) -- (78.37,77.65) ;
\draw [color={rgb, 255:red, 208; green, 2; blue, 27 }  ,draw opacity=1 ][line width=1.5]    (106.37,49.15) -- (106.37,77.65) ;
\draw [color={rgb, 255:red, 208; green, 2; blue, 27 }  ,draw opacity=1 ][line width=1.5]    (134.37,49.15) -- (134.37,77.65) ;
\draw [color={rgb, 255:red, 208; green, 2; blue, 27 }  ,draw opacity=1 ][line width=1.5]    (162.37,49.15) -- (162.37,77.65) ;
\draw [color={rgb, 255:red, 208; green, 2; blue, 27 }  ,draw opacity=1 ][line width=1.5]    (190.37,49.15) -- (190.37,77.65) ;
\draw [color={rgb, 255:red, 208; green, 2; blue, 27 }  ,draw opacity=1 ][line width=1.5]    (78.37,76.65) -- (106.37,76.65) ;
\draw [color={rgb, 255:red, 208; green, 2; blue, 27 }  ,draw opacity=1 ][line width=1.5]    (134.37,76.65) -- (162.37,76.65) ;
\draw [color={rgb, 255:red, 208; green, 2; blue, 27 }  ,draw opacity=1 ][line width=1.5]    (162.37,76.65) -- (190.37,76.65) ;
\draw [color={rgb, 255:red, 208; green, 2; blue, 27 }  ,draw opacity=1 ][line width=1.5]    (134.37,75.15) -- (134.37,103.65) ;
\draw [color={rgb, 255:red, 208; green, 2; blue, 27 }  ,draw opacity=1 ][line width=1.5]    (162.37,75.15) -- (162.37,103.65) ;
\draw [color={rgb, 255:red, 208; green, 2; blue, 27 }  ,draw opacity=1 ][line width=1.5]    (134.37,102.65) -- (162.37,102.65) ;
\draw [color={rgb, 255:red, 208; green, 2; blue, 27 }  ,draw opacity=1 ][line width=1.5]    (50.37,101.15) -- (50.37,129.65) ;
\draw [color={rgb, 255:red, 208; green, 2; blue, 27 }  ,draw opacity=1 ][line width=1.5]    (78.37,101.15) -- (78.37,129.65) ;
\draw [color={rgb, 255:red, 208; green, 2; blue, 27 }  ,draw opacity=1 ][line width=1.5]    (50.37,128.65) -- (78.37,128.65) ;
\draw [color={rgb, 255:red, 208; green, 2; blue, 27 }  ,draw opacity=1 ][line width=1.5]    (78.37,101.15) -- (78.37,129.65) ;
\draw [color={rgb, 255:red, 208; green, 2; blue, 27 }  ,draw opacity=1 ][line width=1.5]    (106.37,101.15) -- (106.37,129.65) ;
\draw [color={rgb, 255:red, 208; green, 2; blue, 27 }  ,draw opacity=1 ][line width=1.5]    (78.37,128.65) -- (106.37,128.65) ;
\draw [color={rgb, 255:red, 208; green, 2; blue, 27 }  ,draw opacity=1 ][line width=1.5]    (78.37,152.15) -- (78.37,180.65) ;
\draw [color={rgb, 255:red, 208; green, 2; blue, 27 }  ,draw opacity=1 ][line width=1.5]    (78.37,179.65) -- (106.37,179.65) ;
\draw [color={rgb, 255:red, 208; green, 2; blue, 27 }  ,draw opacity=1 ][line width=1.5]    (106.37,152.15) -- (106.37,180.65) ;
\draw [color={rgb, 255:red, 208; green, 2; blue, 27 }  ,draw opacity=1 ][line width=1.5]    (134.37,152.15) -- (134.37,180.65) ;
\draw [color={rgb, 255:red, 208; green, 2; blue, 27 }  ,draw opacity=1 ][line width=1.5]    (106.37,179.65) -- (134.37,179.65) ;
\draw [color={rgb, 255:red, 208; green, 2; blue, 27 }  ,draw opacity=1 ][line width=1.5]    (134.37,179.65) -- (162.37,179.65) ;
\draw [color={rgb, 255:red, 208; green, 2; blue, 27 }  ,draw opacity=1 ][line width=1.5]    (162.37,152.15) -- (162.37,180.65) ;
\draw [color={rgb, 255:red, 208; green, 2; blue, 27 }  ,draw opacity=1 ][line width=1.5]    (190.37,152.15) -- (190.37,180.65) ;
\draw [color={rgb, 255:red, 208; green, 2; blue, 27 }  ,draw opacity=1 ][line width=1.5]    (162.37,179.65) -- (190.37,179.65) ;
\draw    (240,129.51) -- (391.33,129.51) ;
\draw    (240,155.38) -- (390.74,155.38) ;
\draw    (245.27,46.15) -- (245.27,187) ;
\draw    (240.59,103.64) -- (390.74,103.64) ;
\draw    (240.59,77.19) -- (390.74,77.19) ;
\draw    (241.17,51.32) -- (392.5,51.32) ;
\draw    (273.37,46.15) -- (273.37,187) ;
\draw    (301.47,45.57) -- (301.47,185.85) ;
\draw    (240.59,181.25) -- (391.33,181.25) ;
\draw    (329.57,45.57) -- (329.57,186.43) ;
\draw    (357.67,45) -- (357.67,185.28) ;
\draw    (385.77,45) -- (385.77,185.28) ;
\draw [color={rgb, 255:red, 208; green, 2; blue, 27 }  ,draw opacity=1 ][line width=1.5]    (273.37,50.15) -- (273.37,78.65) ;
\draw [color={rgb, 255:red, 208; green, 2; blue, 27 }  ,draw opacity=1 ][line width=1.5]    (301.37,50.15) -- (301.37,78.65) ;
\draw [color={rgb, 255:red, 208; green, 2; blue, 27 }  ,draw opacity=1 ][line width=1.5]    (329.37,50.15) -- (329.37,78.65) ;
\draw [color={rgb, 255:red, 208; green, 2; blue, 27 }  ,draw opacity=1 ][line width=1.5]    (357.37,50.15) -- (357.37,78.65) ;
\draw [color={rgb, 255:red, 208; green, 2; blue, 27 }  ,draw opacity=1 ][line width=1.5]    (385.37,50.15) -- (385.37,78.65) ;
\draw [color={rgb, 255:red, 208; green, 2; blue, 27 }  ,draw opacity=1 ][line width=1.5]    (273.37,77.65) -- (301.37,77.65) ;
\draw [color={rgb, 255:red, 208; green, 2; blue, 27 }  ,draw opacity=1 ][line width=1.5]    (329.37,77.65) -- (357.37,77.65) ;
\draw [color={rgb, 255:red, 208; green, 2; blue, 27 }  ,draw opacity=1 ][line width=1.5]    (357.37,77.65) -- (385.37,77.65) ;
\draw [color={rgb, 255:red, 208; green, 2; blue, 27 }  ,draw opacity=1 ][line width=1.5]    (329.37,76.15) -- (329.37,104.65) ;
\draw [color={rgb, 255:red, 208; green, 2; blue, 27 }  ,draw opacity=1 ][line width=1.5]    (357.37,76.15) -- (357.37,104.65) ;
\draw [color={rgb, 255:red, 208; green, 2; blue, 27 }  ,draw opacity=1 ][line width=1.5]    (329.37,103.65) -- (357.37,103.65) ;
\draw [color={rgb, 255:red, 208; green, 2; blue, 27 }  ,draw opacity=1 ][line width=1.5]    (245.37,102.15) -- (245.37,130.65) ;
\draw [color={rgb, 255:red, 208; green, 2; blue, 27 }  ,draw opacity=1 ][line width=1.5]    (273.37,102.15) -- (273.37,130.65) ;
\draw [color={rgb, 255:red, 208; green, 2; blue, 27 }  ,draw opacity=1 ][line width=1.5]    (245.37,129.65) -- (273.37,129.65) ;
\draw [color={rgb, 255:red, 208; green, 2; blue, 27 }  ,draw opacity=1 ][line width=1.5]    (273.37,102.15) -- (273.37,130.65) ;
\draw [color={rgb, 255:red, 208; green, 2; blue, 27 }  ,draw opacity=1 ][line width=1.5]    (301.37,102.15) -- (301.37,130.65) ;
\draw [color={rgb, 255:red, 208; green, 2; blue, 27 }  ,draw opacity=1 ][line width=1.5]    (273.37,129.65) -- (301.37,129.65) ;
\draw [color={rgb, 255:red, 208; green, 2; blue, 27 }  ,draw opacity=1 ][line width=1.5]    (273.37,153.15) -- (273.37,181.65) ;
\draw [color={rgb, 255:red, 208; green, 2; blue, 27 }  ,draw opacity=1 ][line width=1.5]    (273.37,180.65) -- (301.37,180.65) ;
\draw [color={rgb, 255:red, 208; green, 2; blue, 27 }  ,draw opacity=1 ][line width=1.5]    (301.37,153.15) -- (301.37,181.65) ;
\draw [color={rgb, 255:red, 208; green, 2; blue, 27 }  ,draw opacity=1 ][line width=1.5]    (329.37,153.15) -- (329.37,181.65) ;
\draw [color={rgb, 255:red, 208; green, 2; blue, 27 }  ,draw opacity=1 ][line width=1.5]    (301.37,180.65) -- (329.37,180.65) ;
\draw [color={rgb, 255:red, 208; green, 2; blue, 27 }  ,draw opacity=1 ][line width=1.5]    (329.37,180.65) -- (357.37,180.65) ;
\draw [color={rgb, 255:red, 208; green, 2; blue, 27 }  ,draw opacity=1 ][line width=1.5]    (357.37,153.15) -- (357.37,181.65) ;
\draw [color={rgb, 255:red, 208; green, 2; blue, 27 }  ,draw opacity=1 ][line width=1.5]    (385.37,153.15) -- (385.37,181.65) ;
\draw [color={rgb, 255:red, 208; green, 2; blue, 27 }  ,draw opacity=1 ][line width=1.5]    (357.37,180.65) -- (385.37,180.65) ;
\draw    (435,130.51) -- (586.33,130.51) ;
\draw    (435,156.38) -- (585.74,156.38) ;
\draw    (440.27,47.15) -- (440.27,188) ;
\draw    (435.59,104.64) -- (585.74,104.64) ;
\draw    (435.59,78.19) -- (585.74,78.19) ;
\draw    (436.17,52.32) -- (587.5,52.32) ;
\draw    (468.37,47.15) -- (468.37,188) ;
\draw    (496.47,46.57) -- (496.47,186.85) ;
\draw    (435.59,182.25) -- (586.33,182.25) ;
\draw    (524.57,46.57) -- (524.57,187.43) ;
\draw    (552.67,46) -- (552.67,186.28) ;
\draw    (580.77,46) -- (580.77,186.28) ;
\draw [color={rgb, 255:red, 208; green, 2; blue, 27 }  ,draw opacity=1 ][line width=1.5]    (468.37,51.15) -- (468.37,79.65) ;
\draw [color={rgb, 255:red, 208; green, 2; blue, 27 }  ,draw opacity=1 ][line width=1.5]    (496.37,51.15) -- (496.37,79.65) ;
\draw [color={rgb, 255:red, 208; green, 2; blue, 27 }  ,draw opacity=1 ][line width=1.5]    (524.37,51.15) -- (524.37,79.65) ;
\draw [color={rgb, 255:red, 208; green, 2; blue, 27 }  ,draw opacity=1 ][line width=1.5]    (552.37,51.15) -- (552.37,79.65) ;
\draw [color={rgb, 255:red, 208; green, 2; blue, 27 }  ,draw opacity=1 ][line width=1.5]    (580.37,51.15) -- (580.37,79.65) ;
\draw [color={rgb, 255:red, 208; green, 2; blue, 27 }  ,draw opacity=1 ][line width=1.5]    (468.37,78.65) -- (496.37,78.65) ;
\draw [color={rgb, 255:red, 208; green, 2; blue, 27 }  ,draw opacity=1 ][line width=1.5]    (524.37,78.65) -- (552.37,78.65) ;
\draw [color={rgb, 255:red, 208; green, 2; blue, 27 }  ,draw opacity=1 ][line width=1.5]    (552.37,78.65) -- (580.37,78.65) ;
\draw [color={rgb, 255:red, 208; green, 2; blue, 27 }  ,draw opacity=1 ][line width=1.5]    (524.37,77.15) -- (524.37,105.65) ;
\draw [color={rgb, 255:red, 208; green, 2; blue, 27 }  ,draw opacity=1 ][line width=1.5]    (552.37,77.15) -- (552.37,105.65) ;
\draw [color={rgb, 255:red, 208; green, 2; blue, 27 }  ,draw opacity=1 ][line width=1.5]    (524.37,104.65) -- (552.37,104.65) ;
\draw [color={rgb, 255:red, 208; green, 2; blue, 27 }  ,draw opacity=1 ][line width=1.5]    (440.37,103.15) -- (440.37,131.65) ;
\draw [color={rgb, 255:red, 208; green, 2; blue, 27 }  ,draw opacity=1 ][line width=1.5]    (468.37,103.15) -- (468.37,131.65) ;
\draw [color={rgb, 255:red, 208; green, 2; blue, 27 }  ,draw opacity=1 ][line width=1.5]    (440.37,130.65) -- (468.37,130.65) ;
\draw [color={rgb, 255:red, 208; green, 2; blue, 27 }  ,draw opacity=1 ][line width=1.5]    (468.37,103.15) -- (468.37,131.65) ;
\draw [color={rgb, 255:red, 208; green, 2; blue, 27 }  ,draw opacity=1 ][line width=1.5]    (496.37,103.15) -- (496.37,131.65) ;
\draw [color={rgb, 255:red, 208; green, 2; blue, 27 }  ,draw opacity=1 ][line width=1.5]    (468.37,130.65) -- (496.37,130.65) ;
\draw [color={rgb, 255:red, 208; green, 2; blue, 27 }  ,draw opacity=1 ][line width=1.5]    (468.37,154.15) -- (468.37,182.65) ;
\draw [color={rgb, 255:red, 208; green, 2; blue, 27 }  ,draw opacity=1 ][line width=1.5]    (468.37,181.65) -- (496.37,181.65) ;
\draw [color={rgb, 255:red, 208; green, 2; blue, 27 }  ,draw opacity=1 ][line width=1.5]    (496.37,154.15) -- (496.37,182.65) ;
\draw [color={rgb, 255:red, 208; green, 2; blue, 27 }  ,draw opacity=1 ][line width=1.5]    (524.37,154.15) -- (524.37,182.65) ;
\draw [color={rgb, 255:red, 208; green, 2; blue, 27 }  ,draw opacity=1 ][line width=1.5]    (496.37,181.65) -- (524.37,181.65) ;
\draw [color={rgb, 255:red, 208; green, 2; blue, 27 }  ,draw opacity=1 ][line width=1.5]    (524.37,181.65) -- (552.37,181.65) ;
\draw [color={rgb, 255:red, 208; green, 2; blue, 27 }  ,draw opacity=1 ][line width=1.5]    (552.37,154.15) -- (552.37,182.65) ;
\draw [color={rgb, 255:red, 208; green, 2; blue, 27 }  ,draw opacity=1 ][line width=1.5]    (580.37,154.15) -- (580.37,182.65) ;
\draw [color={rgb, 255:red, 208; green, 2; blue, 27 }  ,draw opacity=1 ][line width=1.5]    (552.37,181.65) -- (580.37,181.65) ;
\draw [color={rgb, 255:red, 74; green, 144; blue, 226 }  ,draw opacity=1 ][line width=1.5]    (245.37,76.15) -- (245.37,104.65) ;
\draw [color={rgb, 255:red, 74; green, 144; blue, 226 }  ,draw opacity=1 ][line width=1.5]    (245.37,49.15) -- (245.37,77.65) ;
\draw [color={rgb, 255:red, 74; green, 144; blue, 226 }  ,draw opacity=1 ][line width=1.5]    (245.37,127.15) -- (245.37,155.65) ;
\draw [color={rgb, 255:red, 74; green, 144; blue, 226 }  ,draw opacity=1 ][line width=1.5]    (245.37,154.15) -- (245.37,182.65) ;
\draw [color={rgb, 255:red, 74; green, 144; blue, 226 }  ,draw opacity=1 ][line width=1.5]    (273.37,129.15) -- (273.37,157.65) ;
\draw [color={rgb, 255:red, 74; green, 144; blue, 226 }  ,draw opacity=1 ][line width=1.5]    (301.37,128.15) -- (301.37,156.65) ;
\draw [color={rgb, 255:red, 74; green, 144; blue, 226 }  ,draw opacity=1 ][line width=1.5]    (301.37,76.15) -- (301.37,104.65) ;
\draw [color={rgb, 255:red, 74; green, 144; blue, 226 }  ,draw opacity=1 ][line width=1.5]    (273.37,75.15) -- (273.37,103.65) ;
\draw [color={rgb, 255:red, 74; green, 144; blue, 226 }  ,draw opacity=1 ][line width=1.5]    (329.37,102.15) -- (329.37,130.65) ;
\draw [color={rgb, 255:red, 74; green, 144; blue, 226 }  ,draw opacity=1 ][line width=1.5]    (357.37,102.15) -- (357.37,130.65) ;
\draw [color={rgb, 255:red, 74; green, 144; blue, 226 }  ,draw opacity=1 ][line width=1.5]    (357.37,128.15) -- (357.37,156.65) ;
\draw [color={rgb, 255:red, 74; green, 144; blue, 226 }  ,draw opacity=1 ][line width=1.5]    (329.37,125.15) -- (329.37,153.65) ;
\draw [color={rgb, 255:red, 74; green, 144; blue, 226 }  ,draw opacity=1 ][line width=1.5]    (385.37,127.15) -- (385.37,155.65) ;
\draw [color={rgb, 255:red, 74; green, 144; blue, 226 }  ,draw opacity=1 ][line width=1.5]    (385.37,102.15) -- (385.37,130.65) ;
\draw [color={rgb, 255:red, 74; green, 144; blue, 226 }  ,draw opacity=1 ][line width=1.5]    (385.37,76.15) -- (385.37,104.65) ;
\draw [color={rgb, 255:red, 74; green, 144; blue, 226 }  ,draw opacity=1 ][line width=1.5]    (496.37,127.15) -- (496.37,155.65) ;
\draw (120,203) node    {$a)$};
\draw (320,205) node    {$b)$};
\draw (512,205) node    {$c)$};
\end{tikzpicture}}
%If not, use
%\vspace{5cm}       % Give the correct figure height in cm
%Give a unique label
\caption{a) A subset $\ell$ of $\mathbb{O}_{\Delta}$ and the corresponding
$G_{e}, e \in\ell$, b) the paths in $(\Delta,\mathbb{V}_{\Delta})$ connecting
the graphs $G_{e}, e \in\ell$, c) three polymers associated to $\ell.$}%
\end{figure}

\subsection{Reduction to a polymer gas model}

Given $e\in\mathbb{O},$ let
\begin{equation}
V^{\prime}\left(  e\right)  :=\left\{  x\in\mathbb{Z}\times\delta
\mathbb{Z}:\left(  x_{1},\delta\left(  x_{2}-1\right)  \right)  \in
V_{e}\right\}  \label{V'}%
\end{equation}
and
\begin{equation}
E^{\prime}\left(  e\right)  :=\left\{  e^{\prime}\in\mathbb{V}:e^{\prime
}=\left\{  x,y\right\}  ,x\in V_{e},y\in V^{\prime}\left(  e\right)  \right\}
\ .
\end{equation}
We set $G_{e}:=\left(  V_{e}\cup V^{\prime}\left(  e\right)  ,e\cup E^{\prime
}\left(  e\right)  \right)  \subset\mathbb{L}_{\delta}^{2}.$

We call \emph{polymer} a connected subgraph $R$ of $\mathbb{L}_{\delta}^{2}$
which satisfies the following conditions:

\begin{enumerate}
\item for any $e\in E\left(  R\right)  \cap\mathbb{O},G_{e}\subseteq R;$

\item if $e$ and $e^{\prime}$ are two distinct edges in $E\left(  R\right)
\cap\mathbb{O},$ either $G_{e}\cup G_{e^{\prime}}$ is a connected subgraph of
$\mathbb{L}_{\delta}^{2}$\ or, given a path $\gamma$ connecting $G_{e}$ and
$G_{e^{\prime}},$ for any $e^{\prime\prime}\in E\left(  \gamma\right)
\cap\mathbb{O},G_{e^{\prime\prime}}\subset R.$
\end{enumerate}

Given a polymer $R$ (an example is a connected subgraph of the graph in fig.2
c)) we set $\left\Vert R\right\Vert :=\left\vert E\left(  R\right)
\right\vert .$ Denoting by $\mathfrak{R}$ the set of polymers, $R,R^{\prime
}\in\mathfrak{R}$ are said to be \emph{compatible}, and we write $R\sim
R^{\prime},$ if $V\left(  R\right)  \cap V\left(  R^{\prime}\right)
=\varnothing,$ otherwise are said to be \emph{incompatible} and we write
$R\nsim R^{\prime}.$ Given $\mathcal{R}\subset\mathfrak{R},$ we denote by
$\mathfrak{P}\left(  \mathcal{R}\right)  $ the collection of the subsets of
$\mathcal{R}$ consisting of mutually compatible polymers and by $\mathfrak{P}%
_{0}\left(  \mathcal{R}\right)  :=\left\{  \varrho\in\mathfrak{P}\left(
\mathcal{R}\right)  :\left\Vert \varrho\right\Vert <\infty\right\}  .$ We also
set $\mathfrak{P}:=\mathfrak{P}\left(  \mathfrak{R}\right)  ,\mathfrak{P}%
_{0}:=\mathfrak{P}_{0}\left(  \mathfrak{R}\right)  .$ Given $\mathcal{R}%
\in\mathcal{P}_{f}\left(  \mathfrak{R}\right)  $ and $R\in\mathfrak{R}$ we
write $\mathcal{R}\nsim R$ if there exists $R^{\prime}\in\mathcal{R}$ such
that $R^{\prime}\nsim R.$ Moreover, we call $\mathcal{R}$ a \emph{polymer
cluster} if it cannot be decomposed as a union of $\mathcal{R}_{1}%
,\mathcal{R}_{2}\in\mathcal{P}_{f}\left(  \mathfrak{R}\right)  $ such that
every pair $R_{1}\in\mathcal{R}_{1},R_{2}\in\mathcal{R}_{2}$ is compatible. We
denote by $\mathcal{C}\left(  \mathcal{R}\right)  $ the collection of polymer
clusters in $\mathcal{R}$ and let $\mathcal{C}$ be the collection of polymer
clusters in $\mathfrak{R}.$

Given a finite $\Delta:=\Lambda\times I\subset\mathbb{Z}\times\delta
\mathbb{Z}$ we denote by
\begin{equation}
V_{\Delta}^{+}:=%
%TCIMACRO{\dbigcup \limits_{e\in\mathbb{O}\ :\ V_{e}\subset\partial^{+}\Delta
%}}%
%BeginExpansion
{\displaystyle\bigcup\limits_{e\in\mathbb{O}\ :\ V_{e}\subset\partial
^{+}\Delta}}
%EndExpansion
V^{\prime}\left(  e\right)  \label{V+}%
\end{equation}
and set $\mathfrak{R}_{\Delta}$ the collection of polymers $R\in\mathfrak{R}$
such that:

\begin{itemize}
\item $V\left(  R\right)  \subseteq\Delta\cup V_{\Delta}^{+};$

\item if $V\left(  R\right)  \cap\overline{\partial}\Delta\neq\varnothing$
then either $\partial^{+}\Delta$ or $\partial^{-}\Delta$ or $\overline
{\partial}\Delta=\partial^{+}\Delta\cup\partial^{-}\Delta$ are contained in
$V\left(  R\right)  .$
\end{itemize}

We also set $\mathfrak{P}_{\Delta}:=\mathfrak{P}\left(  \mathfrak{R}_{\Delta
}\right)  .$ Then, for any $\mathcal{R}\subseteq\mathfrak{R}_{\Delta},$ we
define
\begin{equation}
\mathcal{Z}\left(  \mathcal{R},\Phi^{h,\xi}\right)  :=\sum_{\varrho
\in\mathfrak{P}\left(  \mathcal{R}\right)  }\prod\limits_{R\in\varrho}%
\Phi^{h,\xi}\left(  R\right)  \ ,
\end{equation}
where the function $\mathfrak{R}\ni R\longmapsto\Phi^{h,\xi}\left(  R\right)
\in\mathbb{R}^{+}$ is the \emph{activity} of the polymer.

By (\ref{Zd2}),
\begin{align}
Z_{\Delta}\left(  \xi\right)   &  =\sum_{\ell\in\mathcal{P}\left(
\mathbb{O}_{\Delta}\right)  }2^{\left\vert V\left(  \ell\right)  \right\vert
}\int%
%TCIMACRO{\dbigotimes \limits_{x\in V\left(  \ell\right)  }}%
%BeginExpansion
{\displaystyle\bigotimes\limits_{x\in V\left(  \ell\right)  }}
%EndExpansion
\mu^{\delta}\left(  \text{d}\sigma_{x}\right)
%TCIMACRO{\dprod \limits_{\left\{  x,y\right\}  \in\mathbb{V}_{\Delta}\ :\ x\in
%V\left(  \ell\right)  ,y\in V^{\prime}\left(  \ell\right)  }}%
%BeginExpansion
{\displaystyle\prod\limits_{\left\{  x,y\right\}  \in\mathbb{V}_{\Delta
}\ :\ x\in V\left(  \ell\right)  ,y\in V^{\prime}\left(  \ell\right)  }}
%EndExpansion
e^{W_{1}\left(  \sigma_{x},\sigma_{y}\right)  }\times\\
&  \times\prod\limits_{e\in\ell}\mathbf{1}_{e}\left(  \left\{  x,y\right\}
\right)  \left(  e^{W\left(  \sigma_{x},\sigma_{y}\right)  }-1\right)
\times\nonumber\\
&  \times\sum_{\mathfrak{g\in}\mathcal{P}\left(  \Pi\left(  \ell\right)
\right)  }\left(  \frac{1}{2}\right)  ^{\left\vert \Pi\left(  \ell\right)
\backslash\mathfrak{g}\right\vert }\prod\limits_{\gamma\in\mathfrak{g}%
}\mathbf{1}_{end\left(  \gamma\right)  }\left(  \left\{  x,y\right\}  \right)
\sigma_{x}\sigma_{y}\frac{e^{-2h\delta\left\vert E\left(  \gamma\right)
\right\vert }}{2}\times\nonumber\\
&  \times\left[  \left(  1-\mathbf{1}_{\partial^{+}\Delta}\left(  x\right)
\right)  +\mathbf{1}_{\partial^{+}\Delta}\left(  x\right)  \mathbf{1}_{\xi
_{x}^{+}}\left(  \sigma_{x}\right)  \right]  \left[  \left(  1-\mathbf{1}%
_{\partial^{-}\Delta}\left(  y\right)  \right)  +\mathbf{1}_{\partial
^{-}\Delta}\left(  y\right)  \mathbf{1}_{\xi_{y}^{-}}\left(  \sigma
_{y}\right)  \right]  \ .\nonumber
\end{align}
Then, given $\ell\in\mathcal{P}\left(  \mathbb{O}_{\Delta}\right)  $ and
$\mathfrak{g\in}\mathcal{P}\left(  \Pi\left(  \ell\right)  \right)  ,$ the
components of $\varrho\left(  \ell\right)  :=\left(
%TCIMACRO{\dbigcup \limits_{e\in\overline{\ell}}}%
%BeginExpansion
{\displaystyle\bigcup\limits_{e\in\overline{\ell}}}
%EndExpansion
G_{e}\right)  \cup\left(
%TCIMACRO{\dbigcup \limits_{\gamma\in\mathfrak{g}}}%
%BeginExpansion
{\displaystyle\bigcup\limits_{\gamma\in\mathfrak{g}}}
%EndExpansion
\gamma\right)  ,$ with $\overline{\ell}:=\ell\cup\left\{  e\in\mathbb{O}%
\ :\ V_{e}\subset\overline{\partial}\Delta\right\}  ,$ fit the definition of
polymer, hence we can write
\begin{align}
Z_{\Delta}\left(  \xi\right)   &  =\sum_{\ell\in\mathcal{P}\left(
\mathbb{O}_{\Delta}\right)  }\sum_{\varrho\in\mathfrak{P}_{\ell}}%
%TCIMACRO{\dprod \limits_{R\in\varrho}}%
%BeginExpansion
{\displaystyle\prod\limits_{R\in\varrho}}
%EndExpansion
2^{\left\vert U_{R}\right\vert }\int\mu^{\delta}\left(  \text{d}\sigma_{U_{R}%
}\right)
%TCIMACRO{\dprod \limits_{\left\{  x,y\right\}  \in\mathbb{V}_{\Delta}\ :\ x\in
%U_{R},y\in U_{R}^{\prime}}}%
%BeginExpansion
{\displaystyle\prod\limits_{\left\{  x,y\right\}  \in\mathbb{V}_{\Delta
}\ :\ x\in U_{R},y\in U_{R}^{\prime}}}
%EndExpansion
e^{W_{1}\left(  \sigma_{x},\sigma_{y}\right)  }\times\\
&  \times\prod\limits_{e\in E\left(  R\right)  \cap\mathbb{O}_{\Delta}%
}\mathbf{1}_{e}\left(  \left\{  x,y\right\}  \right)  \left(  e^{W\left(
\sigma_{x},\sigma_{y}\right)  }-1\right)  \times\nonumber\\
&  \times\prod\limits_{\gamma\in\mathfrak{g}\left(  R\right)  }\mathbf{1}%
_{end\left(  \gamma\right)  }\left(  \left\{  x,y\right\}  \right)  \sigma
_{x}\sigma_{y}\frac{e^{-2h\delta\left\vert E\left(  \gamma\right)  \right\vert
}}{2}\prod\limits_{\gamma\in\Pi\left(  \ell\right)  \backslash\mathfrak{g}%
\left(  R\right)  }\mathbf{1}_{end\left(  \gamma\right)  }\left(  \left\{
x,y\right\}  \right)  \frac{1}{2}\times\nonumber\\
&  \times\left[  \left(  1-\mathbf{1}_{\partial^{+}\Delta}\left(  x\right)
\right)  +\mathbf{1}_{\partial^{+}\Delta}\left(  x\right)  \mathbf{1}_{\xi
_{x}^{+}}\left(  \sigma_{x}\right)  \right]  \left[  \left(  1-\mathbf{1}%
_{\partial^{-}\Delta}\left(  y\right)  \right)  +\mathbf{1}_{\partial
^{-}\Delta}\left(  y\right)  \mathbf{1}_{\xi_{y}^{-}}\left(  \sigma
_{y}\right)  \right]  \ ,\nonumber
\end{align}
where, for any $\ell\in\mathcal{P}\left(  \mathbb{O}_{\Delta}\right)
,\mathfrak{P}_{\ell}$ is the set of the collections of mutually compatible
polymers which can be realised as union set of $%
%TCIMACRO{\dbigcup \limits_{e\in\overline{\ell}}}%
%BeginExpansion
{\displaystyle\bigcup\limits_{e\in\overline{\ell}}}
%EndExpansion
G_{e}$ with elements of $\Pi\left(  \ell\right)  $ and, for any polymer $R$ in
$\varrho\in\mathfrak{P}_{\ell},U_{R}:=V\left(  \ell\right)  \cap V\left(
R\right)  ,U_{R}^{\prime}:=V^{\prime}\left(  \ell\right)  \cap V\left(
R\right)  $ and $\mathfrak{g}\left(  R\right)  :=\left\{  \gamma\in\Pi\left(
\ell\right)  :\gamma\subset R\right\}  .$

Given $\varrho\in\mathfrak{P}_{\Delta},$ setting $\ell\left(  \varrho\right)
:=E\left(  \varrho\right)  $ and consequently $V\left(  \varrho\right)
:=V\left(  \ell\left(  \rho\right)  \right)  $ and $\Pi\left(  \varrho\right)
:=\Pi\left(  \ell\left(  \rho\right)  \right)  ,Z_{\Delta}\left(  \xi\right)
$ can be rewritten as
\begin{align}
Z_{\Delta}\left(  \xi\right)   &  =\sum_{\varrho\in\mathfrak{P}_{\Delta}}%
\prod\limits_{R\in\varrho}2^{\left\vert U_{R}\right\vert }\int\mu^{\delta
}\left(  \text{d}\sigma_{U_{R}}\right)
%TCIMACRO{\dprod \limits_{\left\{  x,y\right\}  \in\mathbb{V}_{\Delta}\ :\ x\in
%U_{R},y\in U_{R}^{\prime}}}%
%BeginExpansion
{\displaystyle\prod\limits_{\left\{  x,y\right\}  \in\mathbb{V}_{\Delta
}\ :\ x\in U_{R},y\in U_{R}^{\prime}}}
%EndExpansion
e^{W_{1}\left(  \sigma_{x},\sigma_{y}\right)  }\times\\
&  \times\prod\limits_{e\in E\left(  R\right)  \cap\mathbb{O}_{\Delta}%
}\mathbf{1}_{e}\left(  \left\{  x,y\right\}  \right)  \left(  e^{W\left(
\sigma_{x},\sigma_{y}\right)  }-1\right)  \times\nonumber\\
&  \times\left(  \frac{1}{2}\right)  ^{\left\vert \Pi\left(  \varrho\right)
\backslash\mathfrak{g}\left(  R\right)  \right\vert }\prod\limits_{\gamma
\in\mathfrak{g}\left(  R\right)  }\mathbf{1}_{end\left(  \gamma\right)
}\left(  \left\{  x,y\right\}  \right)  \sigma_{x}\sigma_{y}\frac
{e^{-2h\delta\left\vert E\left(  \gamma\right)  \right\vert }}{2}%
\times\nonumber\\
&  \times\left[  \left(  1-\mathbf{1}_{\partial^{+}\Delta}\left(  x\right)
\right)  +\mathbf{1}_{\partial^{+}\Delta}\left(  x\right)  \mathbf{1}_{\xi
_{x}^{+}}\left(  \sigma_{x}\right)  \right]  \left[  \left(  1-\mathbf{1}%
_{\partial^{-}\Delta}\left(  y\right)  \right)  +\mathbf{1}_{\partial
^{-}\Delta}\left(  y\right)  \mathbf{1}_{\xi_{y}^{-}}\left(  \sigma
_{y}\right)  \right]  \ .\nonumber
\end{align}
Hence,
\begin{equation}
Z_{\Delta}\left(  \xi\right)  =\mathcal{Z}\left(  \mathfrak{R}_{\Delta}%
,\Phi^{h,\xi}\right)  =\sum_{\varrho\in\mathfrak{P}_{\Delta}}\prod
\limits_{R\in\varrho}\Phi^{h,\xi}\left(  R\right)  \ , \label{Z=Z}%
\end{equation}
with
\begin{align}
\Phi^{h,\xi}\left(  R\right)   &  :=2^{\left\vert U_{R}\right\vert }\int
\mu^{\delta}\left(  \text{d}\sigma_{U_{R}}\right)
%TCIMACRO{\dprod \limits_{\left\{  x,y\right\}  \in\mathbb{V}_{\Delta}\ :\ x\in
%U_{R},y\in U_{R}^{\prime}}}%
%BeginExpansion
{\displaystyle\prod\limits_{\left\{  x,y\right\}  \in\mathbb{V}_{\Delta
}\ :\ x\in U_{R},y\in U_{R}^{\prime}}}
%EndExpansion
e^{W_{1}\left(  \sigma_{x},\sigma_{y}\right)  }\times\label{FIR}\\
&  \times\prod\limits_{e\in E\left(  R\right)  \cap\mathbb{O}}\mathbf{1}%
_{e}\left(  \left\{  x,y\right\}  \right)  \left(  e^{W\left(  \sigma
_{x},\sigma_{y}\right)  }-1\right)  \times\nonumber\\
&  \times\left(  \frac{1}{2}\right)  ^{\left\vert \Pi\left(  \varrho\right)
\backslash\mathfrak{g}\left(  R\right)  \right\vert }\prod\limits_{\gamma
\in\mathfrak{g}\left(  R\right)  }\mathbf{1}_{end\left(  \gamma\right)
}\left(  \left\{  x,y\right\}  \right)  \sigma_{x}\sigma_{y}\frac
{e^{-2h\delta\left\vert E\left(  \gamma\right)  \right\vert }}{2}%
\times\nonumber\\
&  \times\left[  \left(  1-\mathbf{1}_{\partial^{+}\Delta}\left(  x\right)
\right)  +\mathbf{1}_{\partial^{+}\Delta}\left(  x\right)  \mathbf{1}_{\xi
_{x}^{+}}\left(  \sigma_{x}\right)  \right]  \left[  \left(  1-\mathbf{1}%
_{\partial^{-}\Delta}\left(  y\right)  \right)  +\mathbf{1}_{\partial
^{-}\Delta}\left(  y\right)  \mathbf{1}_{\xi_{y}^{-}}\left(  \sigma
_{y}\right)  \right]  \ .\nonumber
\end{align}

Choosing $\delta=\frac{1}{\sqrt{h}},$ for any $\xi\in\Omega_{\overline
{\partial}\Delta},R\in\mathfrak{R}_{\Delta},$ we have%

\begin{equation}
\Phi^{h,\xi}\left(  R\right)  \leq\left(  e^{\frac{J}{\sqrt{h}}}-1\right)
^{\left\vert E\left(  R\right)  \cap\mathbb{O}\right\vert }e^{-2\sqrt
{h}\left\vert E\left(  R\right)  \cap\mathbb{V}\right\vert }\leq e^{-a\left(
h\right)  \left\Vert R\right\Vert }\ ,
\end{equation}
with
\begin{equation}
e^{-a\left(  h\right)  }:=\max\left\{  \left(  e^{\frac{J}{\sqrt{h}}%
}-1\right)  ,e^{-2\sqrt{h}}\right\}  \ .
\end{equation}

We remark that to deal with b.c.'s that are free on the top and on the bottom
of $\Delta,$ or are periodic in the vertical direction, the definition of the
polymers activity must be changed slightly.

The previous bound implies that the cluster expansion is convergent when $h$
is sufficiently large. Indeed we can choose a constant $c>0$ such that
\begin{equation}
\sum_{R^{\prime}\in\mathfrak{R}_{\Delta}\ :\ R^{\prime}\nsim R}e^{c\left\Vert
R^{\prime}\right\Vert }e^{-a\left(  h\right)  \left\Vert R^{\prime}\right\Vert
}\leq\frac{c}{2}\left\Vert R\right\Vert \ , \label{conv}%
\end{equation}
which is a sufficient condition for the theorem in \cite{KP} to hold.
Therefore, for any $\mathcal{R}\subseteq\mathfrak{R}_{\Delta},$%
\begin{equation}
\log\mathcal{Z}\left(  \mathfrak{R}_{\Delta},\Phi^{h,\xi}\right)
=\sum_{\mathcal{R}^{\prime}\in\mathcal{C}\left(  \mathcal{R}\right)  }%
\hat{\Phi}^{h,\xi}\left(  \mathcal{R}^{\prime}\right)  \label{Z}%
\end{equation}
where, setting $\mathcal{C}_{\Delta}:=\mathcal{C}\left(  \mathfrak{R}_{\Delta
}\right)  ,$ in view of (\ref{conv}),%
\begin{equation}
\mathcal{C}_{\Delta}\ni\mathcal{R\longmapsto}\hat{\Phi}^{h,\xi}\left(
\mathcal{R}\right)  :=\sum_{\mathcal{R}^{\prime}\in\mathcal{P}\left(
\mathcal{R}\right)  }\left(  -1\right)  ^{\left\vert \mathcal{R}\right\vert
-\left\vert \mathcal{R}^{\prime}\right\vert }\log\mathcal{Z}\left(
\mathfrak{R}_{\Delta},\Phi^{h,\xi}\right)
\end{equation}
is such that, $\forall R\in\mathfrak{R}_{\Delta},$%
\begin{equation}
\sum_{\mathcal{R}^{\prime}\in\mathcal{C}_{\Delta}\ :\ \mathcal{R}^{\prime
}\nsim R}\left\vert \hat{\Phi}^{h,\xi}\left(  \mathcal{R}^{\prime}\right)
\right\vert e^{\frac{c}{2}\sum_{R^{\prime}\in\mathcal{R}^{\prime}}\left\Vert
R^{\prime}\right\Vert }\leq\frac{c}{2}\left\Vert R\right\Vert \ . \label{ext1}%
\end{equation}
As already highlighted in the proposition in \cite{KP}, the last bound is the
key ingredient to perform estimates of quantities which can be represented, in
the setup of a polymer gas model, as ratios of partition functions of the form
$\frac{\mathcal{Z}\left(  \mathcal{R},\Phi\right)  }{\mathcal{Z}\left(
\mathcal{R},\Phi^{\prime}\right)  }.$ As will clearly appear in the next
section the proof of Theorem \ref{t1} will indeed rely on estimates of this kind.

Moreover, the bound on the polymer activity also implies that two point
correlation functions decay exponentially with the distance when $h$ is large
with an $h$-dependent decay constant.

\section{Slit box variant of the model}

From now on we set $\delta$ equal to $\frac{1}{\sqrt{h}}.$ Let $\mathbb{H}%
:=\left\{  y\in\mathbb{Z}\times\delta\mathbb{Z}:y_{2}=0\right\}  $ and
$\mathbb{H}^{+}:=\left\{  y\in\mathbb{Z}\times\delta\mathbb{Z}:y_{2}%
>0\right\}  ,$\linebreak$\mathbb{H}^{-}:=\left\{  y\in\mathbb{Z}\times
\delta\mathbb{Z}:y_{2}<0\right\}  .$

Given a finite $\Lambda\subset\mathbb{Z},$ we denote by $\bar{\Lambda
}:=\left\{  \left(  x_{1},0\right)  \in\mathbb{H}:x_{1}\in\Lambda\right\}  $
and set
\begin{equation}
\Lambda_{\pm}:=\left\{  y\in\mathbb{R}^{2}:y=\left(  x_{1},\pm\frac{1}%
{2}\right)  ,x_{1}\in\Lambda\right\}  \ .
\end{equation}

In order to discuss the asymptotic scaling of the entanglement entropy of the
ground state of a block of spins, inspired by \cite{GOS}, we consider a
modified model in which $\mathbb{L}_{\delta}^{2}$ is replaced by the graph
$\mathbb{\bar{L}}_{\delta}^{2}$ in such a way that:

\begin{itemize}
\item each lattice point in $x=\left(  x_{1},0\right)  \in\bar{\Lambda}$ is
replaced by two distinct vertices $x^{+}:=\left(  x_{1},\frac{1}{2}\right)  $
and $x^{-}:=\left(  x_{1},-\frac{1}{2}\right)  ;$

\item each bond $e=\left\{  x,y\right\}  \in\mathbb{E}_{\delta}^{2}$ such that
$x\in\bar{\Lambda},y\in\mathbb{H}^{+}$ is replaced by $\left\{  x^{+}%
,y\right\}  ;$

\item each bond $e=\left\{  x,y\right\}  \in\mathbb{E}_{\delta}^{2}$ such that
$x\in\bar{\Lambda},y\in\mathbb{H}^{-}$ is replaced by $\left\{  x^{-}%
,y\right\}  ;$

\item each bond $e=\left\{  x,y\right\}  \in\mathbb{E}_{\delta}^{2}$ such that
$x,y\in\bar{\Lambda}$ is replaced by the bonds $\left\{  x^{+},y^{+}\right\}
,\left\{  x^{-},y^{-}\right\}  ;$

\item each bond $e=\left\{  x,y\right\}  \in\mathbb{E}_{\delta}^{2}$ such that
$y\in\bar{\Lambda},x\in\partial\bar{\Lambda}\cap\mathbb{H}$ is replaced by the
bonds $\left\{  x,y^{+}\right\}  ,\left\{  x,y^{-}\right\}  .$
\end{itemize}

For any $\beta>0,$ let us set $I^{+}:=\left[  0,\frac{\beta}{2}\right]
\cap\delta\mathbb{Z},I^{-}:=\left[  -\frac{\beta}{2},0\right]  \cap
\delta\mathbb{Z}$ and denote $I:=I^{+}\cup I^{-}.$ Moreover, we set
$\Delta:=\Delta_{+}\cup\Delta_{-},$ where $\Delta_{\pm}:=\Lambda_{\pm}\times
I^{\pm},$ and keep the definitions of $\partial^{\pm}\Delta,\overline
{\partial}\Delta$ and $\partial\Delta$ given in (\ref{b+-}) and (\ref{bDelta}%
). We also keep the definition of $V_{\Delta}^{+}$ given in (\ref{V+}) and
define $\mathbb{\bar{V}}_{\Delta}$ and $\mathbb{\bar{O}}_{\Delta}$ according
to the definitions of $\mathbb{V}_{\Delta}$ and $\mathbb{O}_{\Delta}$ given at
the beginning of Section \ref{Cluster}.

Then, the generic matrix element $\rho_{\Lambda}^{\beta}\left(  \epsilon
^{+},\epsilon^{-}\right)  ,$ with $\epsilon^{\pm}\in\Omega_{\Lambda},$ of the
density operator $\rho_{\Lambda}^{\beta}$ on $\mathcal{H}_{\Lambda}$
associated to the Hamiltonian (\ref{HtfI}) can be rewritten in terms of a
Gibbsian specification $\frac{\nu_{\delta}^{p}\left(  \text{d}\sigma_{\Delta
}\right)  }{\mu^{\delta}\left(  \text{d}\sigma_{\Delta}\right)  }$ for a spin
model defined on $\mathbb{\bar{L}}_{\delta}^{2}$ by a two-body potential
$W_{1}+W$ analogous to that given in (\ref{Gibbs2}) with periodic b.c.'s at
$\overline{\partial}\Delta.$

Let us set $R_{+}:=\bigcup\limits_{e\in\mathbb{\bar{O}}_{\Delta}%
\ :\ V_{e}\subset\Lambda_{+}}G_{e}$ and define $R_{-}$ to be the graph such
that $V\left(  R_{-}\right)  :=\Lambda_{-}\cup\bar{\Lambda}$ and $E\left(
R_{-}\right)  :=\left\{  e\in\mathbb{\bar{O}}_{\Delta}\ :\ V_{e}\subset
\Lambda_{-}\right\}  \cup\left\{  \left\{  x,y\right\}  \in\mathbb{R}^{2}%
:x\in\Lambda_{-},y\in\bar{\Lambda}\right\}  .$ We denote by $\mathfrak{R}$ the
union set of $\left\{  R_{+},R_{-}\right\}  $ with the collection of polymers
$R$ in $\mathbb{\bar{L}}_{\delta}^{2}$ such that $V\left(  R\right)
\subset\Lambda_{+}^{c}.$ Hence, assuming periodic b.c.'s at $\overline
{\partial}\Delta,$ fixed b.c. $\left(  \epsilon^{+},\epsilon^{-}\right)
\in\Omega_{\Lambda_{+}}\times\Omega_{\Lambda_{-}}$ and free b.c.'s at
$\partial\Delta,$ denoting by $\Phi_{\sigma_{\Lambda_{+}}=\epsilon^{+}%
,\sigma_{\Lambda_{-}}=\epsilon^{-}}^{h}$ the activity of the polymers in
\begin{equation}
\mathfrak{R}_{\Delta}:=\left\{  R\in\mathfrak{R}:V\left(  R\right)
\subseteq\Delta\cup\bar{\Lambda}\right\}  \ , \label{RD}%
\end{equation}
by (\ref{Z=Z}), we have
\begin{equation}
\rho_{\Lambda}^{\beta}\left(  \epsilon^{+},\epsilon^{-}\right)  =\frac
{\nu_{\delta}^{p}\left(  \mathbf{1}_{\left\{  \sigma_{\Lambda_{+}}%
=\epsilon^{+},\sigma_{\Lambda_{-}}=\epsilon^{-}\right\}  }\right)  }%
{\nu_{\delta}^{p}\left(  \mathbf{1}_{\left\{  \sigma_{\Lambda_{+}}%
=\sigma_{\Lambda_{-}}\right\}  }\right)  }=\frac{\mathcal{Z}\left(
\mathfrak{R}_{\Delta},\Phi_{\sigma_{\Lambda_{+}}=\epsilon^{+},\sigma
_{\Lambda_{-}}=\epsilon^{-}}^{h}\right)  }{\mathcal{Z}\left(  \mathfrak{R}%
_{\Delta},\Phi_{\sigma_{\Lambda_{+}}=\sigma_{\Lambda_{-}}}^{h}\right)  }%
\end{equation}
with
\begin{equation}
\mathcal{Z}\left(  \mathfrak{R}_{\Delta},\Phi_{\sigma_{\Lambda_{+}}%
=\sigma_{\Lambda_{-}}}^{h}\right)  =\sum_{\epsilon\in\Omega_{\Lambda}%
}\mathcal{Z}\left(  \mathfrak{R}_{\Delta},\Phi_{\sigma_{\Lambda_{+}}%
=\sigma_{\Lambda_{-}}=\epsilon}^{h}\right)  \ .
\end{equation}

\subsection{The reduced density operator}

From now on we will keep our notation as close as possible to that introduced
in \cite{GOS}. Let us set for $m,L\in\mathbb{N},$%
\begin{equation}
{\Lambda}_{m}{:=\{-m,-m+1,\ldots,m+L\}\ ,\ }\Lambda_{0}:=\left\{
0,..,L\right\}
\end{equation}
and $\Lambda_{m}^{\pm}:=\left(  \Lambda_{m}\right)  _{\pm},\Lambda_{\pm
}:=\left(  \Lambda_{0}\right)  _{\pm},\Delta_{m}:=\Delta_{m}^{+}\cup\Delta
_{m}^{-},$ where $\Delta_{m}^{\pm}:=\Lambda_{m}^{\pm}\times I^{\pm}.$
Considering the representation of $\mathcal{H}_{m}:=\mathcal{H}_{\Lambda_{m}}$
as $\mathcal{H}_{m,L}\otimes\mathcal{H}_{L},$ where $\mathcal{H}%
_{m,L}:=\mathcal{H}_{\Lambda_{m}\backslash\Lambda_{0}},\mathcal{H}%
_{L}:=\mathcal{H}_{\Lambda_{0}},$ we denote by $\bar{\rho}_{m}^{L,\beta}$ the
partial trace of $\rho_{m}^{\beta}:=\rho_{\Lambda_{m}}^{\beta}$ w.r.t.
$\mathcal{H}_{m,L}.$ Then, the generic matrix element $\rho_{m}^{L,\beta
}\left(  \epsilon^{+},\epsilon^{-}\right)  :=\rho_{\Lambda_{m}}^{\Lambda
_{0},\beta}\left(  \epsilon^{+},\epsilon^{-}\right)  ,$ with $\epsilon^{\pm
}\in\Omega_{L}:=\Omega_{\Lambda_{0}},$ of the \emph{reduced density operator}
$\rho_{m}^{L,\beta}:=\frac{\bar{\rho}_{m}^{L,\beta}}{tr_{\mathcal{H}%
_{\Lambda_{0}}}\bar{\rho}_{m}^{L,\beta}}$ on $\mathcal{H}_{L}$ writes
\begin{align}
\rho_{m}^{L,\beta}\left(  \epsilon^{+},\epsilon^{-}\right)   &  =\frac
{\nu_{\delta}^{p}\left(  \mathbf{1}_{\left\{  \sigma_{L^{+}}=\epsilon
^{+},\sigma_{L^{-}}=\epsilon^{-}\right\}  }\right)  }{\nu_{\delta}^{p}\left(
\mathbf{1}_{\left\{  \sigma_{L^{+}}=\sigma_{L^{-}}\right\}  }\right)  }\\
&  =\frac{\sum_{\epsilon^{\prime}\in\Omega_{m,L}}\mathcal{Z}\left(
\mathfrak{R}_{m},\Phi_{\sigma_{m,L}=\epsilon^{\prime},\sigma_{\Lambda_{+}%
}=\epsilon^{+},\sigma_{\Lambda_{-}}=\epsilon^{-}}^{h}\right)  }{\sum
_{\epsilon\in\Omega_{L}}\sum_{\epsilon^{\prime}\in\Omega_{m,L}}\mathcal{Z}%
\left(  \mathfrak{R}_{m},\Phi_{\sigma_{m,L}=\epsilon^{\prime},\sigma
_{\Lambda_{+}}=\sigma_{\Lambda_{-}}=\epsilon}^{h}\right)  }\nonumber\\
&  =\frac{\mathcal{Z}\left(  \mathfrak{R}_{m},\Phi_{\sigma_{\Lambda_{+}%
}=\epsilon^{+},\sigma_{\Lambda_{-}}=\epsilon^{-}}^{h}\right)  }{\mathcal{Z}%
\left(  \mathfrak{R}_{m},\Phi_{\sigma_{\Lambda_{+}}=\sigma_{\Lambda_{-}}}%
^{h}\right)  }\ ,\nonumber
\end{align}
where, $\sigma_{L^{\pm}}:=\sigma_{\Lambda_{\pm}},\sigma_{m,L}:=\sigma
_{\Lambda_{m}\backslash\Lambda_{0}},\Omega_{m,L}:=\Omega_{\Lambda
_{m}\backslash\Lambda_{0}}$ and $\mathfrak{R}_{m}:=\mathfrak{R}_{\Delta_{m}}$
with $\mathfrak{R}_{\Delta_{m}}$ defined as in (\ref{RD}).

Since in the limit of $\beta\rightarrow\infty$ and then of $\left\{
\Lambda_{m}\right\}  \uparrow\mathbb{Z}$ the Gibbs measures defined in
(\ref{Gibbs1}) for different b.c.'s converge weakly to the same limit, the
same conclusion holds for the Gibbs measures defined in (\ref{Gibbs2}).
Therefore, assuming b.c. $\xi\in\Omega_{\partial^{+}\Delta_{m}\cup\partial
^{-}\Delta_{m}}$ at $\overline{\partial}\Delta_{m}$ and free b.c.'s at
$\partial\Delta_{m}\backslash\overline{\partial}\Delta_{m}$ for any event
$\left\{  \cdot\right\}  \in\mathfrak{S}_{\Lambda_{+}\cup\Lambda_{-}},$ we
set
\begin{equation}
\phi_{m,\beta}\left\{  \cdot\right\}  :=\frac{\nu_{\delta}^{\xi}\left(
\mathbf{1}_{\left\{  \cdot\right\}  }\right)  }{\nu_{\delta}^{\xi}\left(
\mathbf{1}_{\Omega_{\Lambda_{+}\cup\Lambda_{-}}}\right)  }%
\end{equation}
and define
\begin{align}
\tilde{\rho}_{m}^{L,\beta}\left(  \epsilon^{+},\epsilon^{-}\right)   &
:=\frac{\phi_{m,\beta}\left\{  \sigma_{L^{+}}=\epsilon^{+},\sigma_{L^{-}%
}=\epsilon^{-}\right\}  }{\phi_{m,\beta}\left\{  \sigma_{L^{+}}=\sigma_{L^{-}%
}\right\}  }\\
&  =\frac{\mathcal{Z}\left(  \mathfrak{R}_{m},\Phi_{\sigma_{\Lambda_{+}%
}=\epsilon^{+},\sigma_{\Lambda_{-}}=\epsilon^{-}}^{h,\xi}\right)
}{\mathcal{Z}\left(  \mathfrak{R}_{m},\Phi_{\sigma_{\Lambda_{+}}%
=\sigma_{\Lambda_{-}}}^{h,\xi}\right)  }\ .\nonumber
\end{align}
where here $\mathfrak{R}_{m}$ is the set of polymers $R$ in $\mathfrak{R}$
with $V\left(  R\right)  \subseteq\Delta_{m}\cup\bar{\Lambda}\cup
V_{\Delta_{m}}^{+}$ such that, if $V\left(  R\right)  \cap\overline{\partial
}\Delta\neq\varnothing,$ then $V\left(  R\right)  $ contains either
$\partial^{+}\Delta$ or $\partial^{-}\Delta$ or both. To simplify the
notation, in the following we will also set $\mathfrak{P}_{m}:=\mathfrak{P}%
_{\Delta_{m}}$ and $\mathcal{C}_{m}:=\mathcal{C}_{\Delta_{m}}.$

\begin{lemma}
There exists a positive value of the external magnetic field $h^{\ast}$ such
that, for any $h>h^{\ast}$ and $L,m\in\mathbb{N},$ uniformly in $\beta
>0,\epsilon^{+},\epsilon^{-}\in\Omega_{L}$ and in the b.c. $\xi\in
\Omega_{\overline{\partial}\Delta_{m}},$%
\begin{equation}
e^{-\psi\left(  c\right)  }\leq\frac{\phi_{m,\beta}\left\{  \sigma_{L^{+}%
}=\epsilon^{+},\sigma_{L^{-}}=\epsilon^{-}\right\}  }{\phi_{m,\beta}\left\{
\sigma_{L^{+}}=\epsilon^{+}\right\}  \phi_{m,\beta}\left\{  \sigma_{L^{-}%
}=\epsilon^{-}\right\}  }\leq e^{\psi\left(  c\right)  }\ , \label{L1}%
\end{equation}
where $\psi\left(  c\right)  :=\frac{8}{1-e^{-\frac{c}{2}}}$ with $c$ the
constant appearing in (\ref{ext1}).
\end{lemma}

\begin{proof}
We proceed as in the proof of the proposition in \cite{KP}. Since by
(\ref{FIR}) for any $R\in\mathfrak{R}_{m}$ compatible with $R_{+}$ and
$R_{-},\Phi_{\sigma_{\Lambda_{+}}=\epsilon^{+},\sigma_{\Lambda_{-}}%
=\epsilon^{-}}^{h,\xi}\left(  R\right)  =\Phi^{h,\xi}\left(  R\right)  ,$ by
(\ref{Z}) we have
\begin{align}
\phi_{m,\beta}\left\{  \sigma_{L^{+}}=\epsilon^{+},\sigma_{L^{-}}=\epsilon
^{-}\right\}   &  =\frac{\mathcal{Z}\left(  \mathfrak{R}_{m},\Phi
_{\sigma_{\Lambda_{+}}=\epsilon^{+},\sigma_{\Lambda_{-}}=\epsilon^{-}}^{h,\xi
}\right)  }{\mathcal{Z}\left(  \mathfrak{R}_{m},\Phi^{h,\xi}\right)  }\\
&  =\exp\left\{  \sum_{\mathcal{R}\in\mathcal{C}_{m}\ :\ \mathcal{R}\nsim
R_{+}\cup R_{-}}\hat{\Phi}_{\sigma_{L^{+}}=\epsilon^{+},\sigma_{L^{-}%
}=\epsilon^{-}}^{h,\xi}\left(  \mathcal{R}\right)  -\hat{\Phi}^{h,\xi}\left(
\mathcal{R}\right)  \right\}  \ .\nonumber
\end{align}
Moreover, because if $R\in\mathfrak{R}_{m}$ is compatible with $R_{+}%
,\Phi_{\sigma_{L^{+}}=\epsilon^{+}}^{h,\xi}\left(  R\right)  =\Phi^{h,\xi
}\left(  R\right)  $ and analogously if $R\in\mathfrak{R}_{m}$ is compatible
with $R_{-},\Phi_{\sigma_{L^{-}}=\epsilon^{-}}^{h,\xi}\left(  R\right)
=\Phi^{h,\xi}\left(  R\right)  ,$%
\begin{align}
\phi_{m,\beta}\left\{  \sigma_{L^{\pm}}=\epsilon^{\pm}\right\}   &
=\frac{\sum_{\epsilon^{\mp}\in\Omega_{L}}\mathcal{Z}\left(  \mathfrak{R}%
_{m},\Phi_{\sigma_{L^{\pm}}=\epsilon^{\pm},\sigma_{L^{\mp}}=\epsilon^{\mp}%
}^{h,\xi}\right)  }{\mathcal{Z}\left(  \mathfrak{R}_{m},\Phi^{h,\xi}\right)
}\\
&  =\frac{\mathcal{Z}\left(  \mathfrak{R}_{m},\Phi_{\sigma_{L^{\pm}}%
=\epsilon^{\pm}}^{h,\xi}\right)  }{\mathcal{Z}\left(  \mathfrak{R}_{m}%
,\Phi^{h,\xi}\right)  }\nonumber\\
&  =\exp\left\{  \sum_{\mathcal{R}\in\mathcal{C}_{m}\ :\ \mathcal{R}\nsim
R_{\pm}}\hat{\Phi}_{\sigma_{L^{\pm}}=\epsilon^{\pm}}^{h,\xi}\left(
\mathcal{R}\right)  -\hat{\Phi}^{h,\xi}\left(  \mathcal{R}\right)  \right\}
\ .\nonumber
\end{align}
Therefore, setting, for any $\mathcal{R}\in\mathcal{C}_{m},V\left(
\mathcal{R}\right)  :=\bigcup\limits_{R\in\mathcal{R}}V\left(  R\right)  $
and
\begin{align}
\bar{\Phi}_{\sigma_{L^{+}}=\epsilon^{+},\sigma_{L^{-}}=\epsilon^{-}}^{h,\xi
}\left(  \mathcal{R}\right)   &  :=\max\left\{  \left\vert \hat{\Phi}%
_{\sigma_{L^{+}}=\epsilon^{+},\sigma_{L^{-}}=\epsilon^{-}}^{h,\xi}\left(
\mathcal{R}\right)  \right\vert ,\left\vert \hat{\Phi}^{h,\xi}\left(
\mathcal{R}\right)  \right\vert ,\right.  \\
&  \left.  \left\vert \hat{\Phi}_{\sigma_{L^{+}}=\epsilon^{+},\sigma_{L^{-}%
}=\epsilon^{-}}^{h,\xi}\left(  \mathcal{R}\right)  \right\vert ,\left\vert
\hat{\Phi}_{\sigma_{L^{+}}=\epsilon^{+},\sigma_{L^{-}}=\epsilon^{-}}^{h,\xi
}\left(  \mathcal{R}\right)  \right\vert \right\}  \ ,\nonumber
\end{align}
by (\ref{ext1}),
\begin{gather}
\frac{\phi_{m,\beta}\left\{  \sigma_{L^{+}}=\epsilon^{+},\sigma_{L^{-}%
}=\epsilon^{-}\right\}  }{\phi_{m,\beta}\left\{  \sigma_{L^{+}}=\epsilon
^{+}\right\}  \phi_{m,\beta}\left\{  \sigma_{L^{-}}=\epsilon^{-}\right\}  }=\\
\exp\left\{  \sum_{\mathcal{R}\in\mathcal{C}_{m}\ :\ \mathcal{R}\nsim
R_{+},\mathcal{R}\nsim R_{-}}\hat{\Phi}_{\sigma_{L^{+}}=\epsilon^{+}%
,\sigma_{L^{-}}=\epsilon^{-}}^{h,\xi}\left(  \mathcal{R}\right)  -\hat{\Phi
}_{\sigma_{L^{+}}=\epsilon^{+}}^{h,\xi}\left(  \mathcal{R}\right)  -\hat{\Phi
}_{\sigma_{L^{+}}=\epsilon^{-}}^{h,\xi}\left(  \mathcal{R}\right)  +\hat{\Phi
}^{h,\xi}\left(  \mathcal{R}\right)  \right\}  \nonumber\\
\leq\exp\left\{  4\sum_{\mathcal{R}\in\mathcal{C}_{m}\ :\ \mathcal{R}\nsim
R_{+},\mathcal{R}\nsim R_{-}}\bar{\Phi}_{\sigma_{L^{+}}=\epsilon^{+}%
,\sigma_{_{L^{-}}}=\epsilon^{-}}^{h,\xi}\left(  \mathcal{R}\right)  \right\}
\nonumber\\
\leq\exp\left\{  4\sum_{x\in\Lambda_{+}}e^{-\frac{c}{2}\inf_{\mathcal{R}%
\in\mathcal{C}_{m}\ :\ \mathcal{R}\nsim R_{-},V\left(  \mathcal{R}\right)  \ni
x}\sum_{R\in\mathcal{R}}\left\Vert R\right\Vert }\times\right.  \nonumber\\
\left.  \times\sum_{\mathcal{R}\in\mathcal{C}_{m}\ :\ \mathcal{R}\nsim
R_{-},V\left(  \mathcal{R}\right)  \ni x}\bar{\Phi}_{\sigma_{L^{+}}%
=\epsilon^{+},\sigma_{_{L^{-}}}=\epsilon^{-}}^{h,\xi}\left(  \mathcal{R}%
\right)  e^{\frac{c}{2}\sum_{R\in\mathcal{R}}\left\Vert R\right\Vert
}\right\}  \nonumber\\
\leq\exp\left\{  2c\left\Vert R_{-}\right\Vert e^{-\frac{c}{2}\left\Vert
R_{-}\right\Vert }\sum_{k=0}^{L}e^{-\frac{c}{2}k\wedge\left(  L-k\right)
}\right\}  \leq e^{\frac{8}{1-e^{-\frac{c}{2}}}}\ .\nonumber
\end{gather}
The lower bound in (\ref{L1}) follows from the estimate
\begin{equation}
\frac{\phi_{m,\beta}\left\{  \sigma_{L^{+}}=\epsilon^{+},\sigma_{L^{-}%
}=\epsilon^{-}\right\}  }{\phi_{m,\beta}\left\{  \sigma_{L^{+}}=\epsilon
^{+}\right\}  \phi_{m,\beta}\left\{  \sigma_{L^{-}}=\epsilon^{-}\right\}
}\geq\exp-4\left\{  \sum_{\mathcal{R}\in\mathcal{C}_{\Delta_{m}}%
\ :\ \mathcal{R}\nsim R_{+},\mathcal{R}\nsim R_{-}}\bar{\Phi}_{\sigma
_{\sigma_{L^{+}}=\epsilon^{+},\sigma_{_{L^{-}}}=\epsilon^{-}}}^{h,\xi}\left(
\mathcal{R}\right)  \right\}  \ .
\end{equation}

\end{proof}

\begin{lemma}
There exists a positive value of the external magnetic field $h^{\ast}$ such
that, for any $h>h^{\ast}$ and any $L,m\in\mathbb{N},$ uniformly in
$\beta>0,\epsilon^{+},\epsilon^{-}\in\Omega_{L}$ and in the b.c.'s $\xi
\in\Omega_{\overline{\partial}\Delta_{n}}$ and $\eta_{n}\in\mathcal{S}%
_{\Delta_{n}^{c}\backslash\overline{\partial}\Delta_{n}},$%
\begin{equation}
e^{-\frac{\psi\left(  c\right)  }{2}e^{-cm}}\leq\frac{\phi_{n,\beta}\left\{
\sigma_{L^{+}}=\epsilon^{+},\sigma_{L^{-}}=\epsilon^{-}\right\}  }%
{\phi_{m,\beta}\left\{  \sigma_{L^{+}}=\epsilon^{+},\sigma_{L^{-}}%
=\epsilon^{-}\right\}  }\leq e^{\frac{\psi\left(  c\right)  }{2}e^{-cm}}\ .
\label{L2}%
\end{equation}

\end{lemma}

\begin{proof}
Let us assume b.c. $\xi\in\Omega_{\partial^{+}\Delta_{n}\cup\partial^{-}%
\Delta_{n}}$ at $\overline{\partial}\Delta_{n}$ and free b.c.'s at
$\partial\Delta_{n}\backslash\overline{\partial}\Delta_{n}.$ The proof of
(\ref{L2}) for more general b.c.'s will follow directly from the one carried
out for this case since a change in the b.c.'s affects only the definition of
the polymers activity in the cluster expansion. As in the proof of the
preceding Lemma we have
\begin{gather}
\frac{\phi_{n,\beta}\left\{  \sigma_{L^{+}}=\epsilon^{+},\sigma_{L^{-}%
}=\epsilon^{-}\right\}  }{\phi_{m,\beta}\left\{  \sigma_{L^{+}}=\epsilon
^{+},\sigma_{L^{-}}=\epsilon^{-}\right\}  }=\\
\frac{\exp\left\{  \sum_{\mathcal{R}\in\mathcal{C}_{n}\ :\ \mathcal{R}%
\nsim\left(  R_{+}\cup R_{-}\right)  }\hat{\Phi}_{\sigma_{L^{+}}=\epsilon
^{+},\sigma_{L^{-}}=\epsilon^{-}}^{h,\xi}\left(  \mathcal{R}\right)
-\hat{\Phi}^{h,\xi}\left(  \mathcal{R}\right)  \right\}  }{\exp\left\{
\sum_{\mathcal{R}\in\mathcal{C}_{m}\ :\ \mathcal{R}\nsim\left(  R_{+}\cup
R_{-}\right)  }\hat{\Phi}_{\sigma_{L^{+}}=\epsilon^{+},\sigma_{L^{-}}%
=\epsilon^{-}}^{h,\xi}\left(  \mathcal{R}\right)  -\hat{\Phi}^{h,\xi}\left(
\mathcal{R}\right)  \right\}  }\ .\nonumber
\end{gather}
Since $\mathfrak{R}_{n}\supset\mathfrak{R}_{m},$ setting
\begin{equation}
\bar{\Phi}_{\sigma_{L^{+}}=\epsilon^{+},\sigma_{L^{-}}=\epsilon^{-}}^{h,\xi
}\left(  \mathcal{R}\right)  :=\max\left\{  \left\vert \hat{\Phi}%
_{\sigma_{L^{+}}=\epsilon^{+},\sigma_{L^{-}}=\epsilon^{-}}^{h,\xi}\left(
\mathcal{R}\right)  \right\vert ,\left\vert \hat{\Phi}^{h,\xi}\left(
\mathcal{R}\right)  \right\vert \right\}  \ ,
\end{equation}
we have
\begin{align}
\frac{\phi_{n,\beta}\left\{  \sigma_{L^{+}}=\epsilon^{+},\sigma_{L^{-}%
}=\epsilon^{-}\right\}  }{\phi_{m,\beta}\left\{  \sigma_{L^{+}}=\epsilon
^{+},\sigma_{L^{-}}=\epsilon^{-}\right\}  }  &  \leq\exp\left\{
2\sum_{\mathcal{R}\in\mathcal{C}_{n}\ :\ V\left(  \mathcal{R}\right)
\cap\left(  \Delta_{n}\bigtriangleup\Delta_{m}\right)  \neq\varnothing
,\mathcal{R}\nsim R_{+},\mathcal{R}\nsim R_{-}}\bar{\Phi}_{\sigma_{L^{+}%
}=\epsilon^{+},\sigma_{L^{-}}=\epsilon^{-}}^{h,\xi}\left(  \mathcal{R}\right)
\right\} \\
&  \leq\exp\left\{  2\sum_{x\in\Lambda_{+}}e^{-\frac{c}{2}\inf_{\mathcal{R}%
\in\mathcal{C}_{n}\ :\ \mathcal{R}\nsim R_{-},V\left(  \mathcal{R}\right)
\cap\left(  \Delta_{n}\bigtriangleup\Delta_{m}\right)  \neq\varnothing
,V\left(  \mathcal{R}\right)  \ni x}\sum_{R\in\mathcal{R}}\left\Vert
R\right\Vert }\times\right. \nonumber\\
&  \left.  \times\sum_{\mathcal{R}\in\mathcal{C}_{n}\ :\ \mathcal{R}\nsim
R_{-},V\left(  \mathcal{R}\right)  \cap\left(  \Delta_{n}\bigtriangleup
\Delta_{m}\right)  \neq\varnothing,V\left(  \mathcal{R}\right)  \ni x}%
\bar{\Phi}_{\sigma_{L^{+}}=\epsilon^{+},\sigma_{L^{-}}=\epsilon^{-}}^{h,\xi
}\left(  \mathcal{R}\right)  e^{\frac{c}{2}\sum_{R\in\mathcal{R}}\left\Vert
R\right\Vert }\right\} \nonumber\\
&  \leq\exp\left\{  c\left\Vert R_{-}\right\Vert e^{-c\left(  \frac{\left\Vert
R_{-}\right\Vert }{2}+m\right)  }\sum_{k=0}^{L}e^{-\frac{c}{2}k\wedge\left(
L-k\right)  }\right\}  \leq e^{\frac{4}{1-e^{-\frac{c}{2}}}e^{-cm}%
}\ .\nonumber
\end{align}
The lower bound in (\ref{L2}) follows from the estimate
\begin{equation}
\frac{\phi_{n,\beta}\left\{  \sigma_{L^{+}}=\epsilon^{+},\sigma_{L^{-}%
}=\epsilon^{-}\right\}  }{\phi_{m,\beta}\left\{  \sigma_{L^{+}}=\epsilon
^{+},\sigma_{L^{-}}=\epsilon^{-}\right\}  }\geq\exp\left\{  -2\sum
_{\mathcal{R}\in\mathcal{C}_{n}\ :\ V\left(  \mathcal{R}\right)  \cap\left(
\Delta_{n}\bigtriangleup\Delta_{m}\right)  \neq\varnothing,\mathcal{R}\nsim
R_{+},\mathcal{R}\nsim R_{-}}\bar{\Phi}_{\sigma_{L^{+}}=\epsilon^{+}%
,\sigma_{L^{-}}=\epsilon^{-}}^{h,\xi}\left(  \mathcal{R}\right)  \right\}  \ .
\end{equation}

\end{proof}

\begin{lemma}
There exists a positive value of the external magnetic field $h^{\ast}$ such
that, for any $h>h^{\ast}$ and any $L,m\in\mathbb{N},$ uniformly in
$\beta>0,\epsilon^{+},\epsilon^{-}\in\Omega_{L}$ and in the b.c.'s $\xi
\in\Omega_{\overline{\partial}\Delta_{n}}$ and $\eta_{n}\in\mathcal{S}%
_{\Delta_{n}^{c}\backslash\overline{\partial}\Delta_{n}},$%
\begin{equation}
\left\vert \frac{\phi_{n,\beta}\left\{  \sigma_{L^{+}}=\sigma_{L^{-}}\right\}
}{\phi_{m,\beta}\left\{  \sigma_{L^{+}}=\sigma_{L^{-}}\right\}  }-\frac
{\phi_{n,\beta}\left\{  \sigma_{L^{+}}=\epsilon^{+},\sigma_{L^{-}}%
=\epsilon^{-}\right\}  }{\phi_{m,\beta}\left\{  \sigma_{L^{+}}=\epsilon
^{+},\sigma_{L^{-}}=\epsilon^{-}\right\}  }\right\vert \leq e^{\frac
{\psi\left(  c\right)  }{2}e^{-cm}}\frac{\psi\left(  c\right)  }{2}e^{-cm}\ .
\label{L3}%
\end{equation}

\end{lemma}

\begin{proof}
Proceeding as in the proof of the previous result, as well as in the proof of
the statement (iii) in the thesis of the Proposition in \cite{KP},
\begin{align}
&  \left\vert \frac{\phi_{n,\beta}\left\{  \sigma_{L^{+}}=\sigma_{L^{-}%
}\right\}  }{\phi_{m,\beta}\left\{  \sigma_{L^{+}}=\sigma_{L^{-}}\right\}
}-\frac{\phi_{n,\beta}\left\{  \sigma_{L^{+}}=\epsilon^{+},\sigma_{L^{-}%
}=\epsilon^{-}\right\}  }{\phi_{m,\beta}\left\{  \sigma_{L^{+}}=\epsilon
^{+},\sigma_{L^{-}}=\epsilon^{-}\right\}  }\right\vert \nonumber\\
&  =\left\vert \exp\sum_{\mathcal{R}\in\mathcal{C}_{n}\ :\ V\left(
\mathcal{R}\right)  \cap\left(  \Delta_{n}\bigtriangleup\Delta_{m}\right)
\neq\varnothing,\mathcal{R}\nsim R_{+},\mathcal{R}\nsim R_{-}}\hat{\Phi
}_{\sigma_{L^{+}}=\sigma_{L^{-}}}^{h,\xi}\left(  \mathcal{R}\right)
-\hat{\Phi}^{h,\xi}\left(  \mathcal{R}\right)  \right. \\
&  \left.  -\exp\sum_{\mathcal{R}\in\mathcal{C}_{n}\ :\ V\left(
\mathcal{R}\right)  \cap\left(  \Delta_{n}\bigtriangleup\Delta_{m}\right)
\neq\varnothing,\mathcal{R}\nsim R_{+},\mathcal{R}\nsim R_{-}}\hat{\Phi
}_{\sigma_{L^{+}}=\epsilon^{+},\sigma_{L^{-}}=\epsilon^{-}}^{h,\xi}\left(
\mathcal{R}\right)  -\hat{\Phi}^{h,\xi}\left(  \mathcal{R}\right)  \right\vert
\\
&  \leq e^{\frac{\psi\left(  c\right)  }{2}e^{-cm}}\left\vert \sum
_{\mathcal{R}\in\mathcal{C}_{n}\ :\ V\left(  \mathcal{R}\right)  \cap\left(
\Delta_{n}\bigtriangleup\Delta_{m}\right)  \neq\varnothing,\mathcal{R}\nsim
R_{+},\mathcal{R}\nsim R_{-}}\hat{\Phi}_{\sigma_{L^{+}}=\sigma_{L^{-}}}%
^{h,\xi}\left(  \mathcal{R}\right)  -\hat{\Phi}_{\sigma_{L^{+}}=\epsilon
^{+},\sigma_{L^{-}}=\epsilon^{-}}^{h,\xi}\left(  \mathcal{R}\right)
\right\vert \nonumber\\
&  \leq e^{\frac{\psi\left(  c\right)  }{2}e^{-cm}}\sum_{\mathcal{R}%
\in\mathcal{C}_{n}\ :\ V\left(  \mathcal{R}\right)  \cap\left(  \Delta
_{n}\bigtriangleup\Delta_{m}\right)  \neq\varnothing,\mathcal{R}\nsim
R_{+},\mathcal{R}\nsim R_{-}}\left(  \left\vert \hat{\Phi}_{\sigma_{L^{+}%
}=\sigma_{L^{-}}}^{h,\xi}\left(  \mathcal{R}\right)  \right\vert +\left\vert
\hat{\Phi}_{\sigma_{L^{+}}=\epsilon^{+},\sigma_{L^{-}}=\epsilon^{-}}^{h,\xi
}\left(  \mathcal{R}\right)  \right\vert \right) \nonumber\\
&  \leq e^{\frac{\psi\left(  c\right)  }{2}e^{-cm}}\frac{\psi\left(  c\right)
}{2}e^{-cm}\ .\nonumber
\end{align}
\qquad
\end{proof}

\subsection{Entanglement entropy\label{Entanglement}}

In order to prove Theorem \ref{t1} we follow the same strategy of the proof of
Theorem 2.8 in \cite{GOS} to which we refer the reader for the details of the computations.

As a matter of fact, it follows from the estimate (\ref{mext}) given
below\ that there exist $C:=C\left(  c\right)  ,C^{\prime}:=C^{\prime}\left(
c\right)  >0,$ such that, for any $k\geq K:=\left\lceil C^{-1}\ln C^{\prime
}\right\rceil ,$ the norm of $\rho_{k+1}^{L}-\rho_{k}^{L}$ is bounded by
$C^{\prime}e^{-C\left(  k-K\right)  }.$ Therefore, denoting by $\mathbf{d}$
the smallest value between the dimension of $\mathcal{H}_{L}$ and that of
$\mathcal{H}_{m,L},$ since $S\left(  \rho_{m}^{L}\right)  =-\sum
_{i=1}^{\mathbf{d}}\alpha_{i}\left(  \rho_{m}^{L}\right)  \log\alpha
_{i}\left(  \rho_{m}^{L}\right)  ,$ where $\left\{  \alpha_{i}\left(  \rho
_{m}^{L}\right)  \right\}  _{i=1}^{\mathbf{d}}$ is the vector of the
eigenvalues of $\rho_{m}^{L}$ arranged in decreasing order, if $2\leq m\leq
K,$ we get that $S\left(  \rho_{m}^{L}\right)  $ is smaller than $2K.$ On the
other hand, if $K\geq m,$ iterating the bound of $\max_{i\geq1}\left\vert
\alpha_{i}\left(  \rho_{K+r+1}^{L}\right)  -\alpha_{i}\left(  \rho_{K+r}%
^{L}\right)  \right\vert \leq C^{\prime}e^{-Cr},r\geq0,$ one can prove that
there exist $C_{0}:=C_{0}\left(  c\right)  >0,\iota:=\iota\left(  c\right)
>2$ such that $\alpha_{i}\left(  \rho_{m}^{L}\right)  \leq C_{0}+\alpha
_{i}\left(  \rho_{K}^{L}\right)  ,$ for $i\leq2^{2K},$ and $\alpha_{i}\left(
\rho_{m}^{L}\right)  \leq\frac{C_{0}}{i^{\iota}}$ for $i>2^{K},$ which leads
to the bound $S\left(  \rho_{m}^{L}\right)  \leq C_{1}K,C_{1}:=C_{1}\left(
c\right)  >0.$

\begin{proposition}
There exists a positive value of the external magnetic field $h^{\ast}$ such
that, for any $h>h^{\ast}$ and for any $L,m,n\in\mathbb{N}$ such that $m<n,$%
\begin{equation}
\left\Vert \tilde{\rho}_{m}^{L}-\tilde{\rho}_{n}^{L}\right\Vert \leq
e^{\frac{\psi\left(  c\right)  }{2}\left(  e^{-cm}+6\right)  }\frac
{\psi\left(  c\right)  }{2}e^{-cm}\ . \label{mext}%
\end{equation}

\end{proposition}

\begin{proof}
Proceeding as in the proof of Theorem 2.2 in \cite{GOS}, we are reduced to
estimate the following quantity
\begin{equation}
\sum_{\epsilon^{+},\epsilon^{-}\in\Omega_{L}}b\left(  \epsilon^{+}\right)
b\left(  \epsilon^{-}\right)  \left\vert \tilde{\rho}_{m}^{L,\beta}\left(
\epsilon^{+},\epsilon^{-}\right)  -\tilde{\rho}_{n}^{L,\beta}\left(
\epsilon^{+},\epsilon^{-}\right)  \right\vert \ , \label{b1}%
\end{equation}
for any real-valued positive function $b$ on $\Omega_{L}$ such that
$\sum_{\epsilon\in\Omega_{L}}b^{2}\left(  \epsilon\right)  =1.$ But
\begin{align}
\left\vert \tilde{\rho}_{m}^{L,\beta}\left(  \epsilon^{+},\epsilon^{-}\right)
-\tilde{\rho}_{n}^{L,\beta}\left(  \epsilon^{+},\epsilon^{-}\right)
\right\vert  &  \leq\left\vert \frac{\phi_{m,\beta}\left\{  \sigma_{L^{+}%
}=\epsilon^{+},\sigma_{L^{-}}=\epsilon^{-}\right\}  }{\phi_{m,\beta}\left\{
\sigma_{L^{+}}=\sigma_{L^{-}}\right\}  }-\frac{\phi_{n,\beta}\left\{
\sigma_{L^{+}}=\epsilon^{+},\sigma_{L^{-}}=\epsilon^{-}\right\}  }%
{\phi_{n,\beta}\left\{  \sigma_{L^{+}}=\sigma_{L^{-}}\right\}  }\right\vert \\
&  =\frac{\phi_{m,\beta}\left\{  \sigma_{L^{+}}=\epsilon^{+},\sigma_{L^{-}%
}=\epsilon^{-}\right\}  }{\phi_{n,\beta}\left\{  \sigma_{L^{+}}=\sigma_{L^{-}%
}\right\}  }\times\nonumber\\
&  \times\left\vert \frac{\phi_{n,\beta}\left\{  \sigma_{L^{+}}=\sigma_{L^{-}%
}\right\}  }{\phi_{m,\beta}\left\{  \sigma_{L^{+}}=\sigma_{L^{-}}\right\}
}-\frac{\phi_{n,\beta}\left\{  \sigma_{L^{+}}=\epsilon^{+},\sigma_{L^{-}%
}=\epsilon^{-}\right\}  }{\phi_{m,\beta}\left\{  \sigma_{L^{+}}=\epsilon
^{+},\sigma_{L^{-}}=\epsilon^{-}\right\}  }\right\vert \ .\nonumber
\end{align}
Moreover, by (\ref{L1}),
\begin{align}
\phi_{m,\beta}\left\{  \sigma_{L^{+}}=\epsilon^{+},\sigma_{L^{-}}=\epsilon
^{-}\right\}   &  =\frac{\phi_{m,\beta}\left\{  \sigma_{L^{+}}=\epsilon
^{+},\sigma_{L^{-}}=\epsilon^{-}\right\}  }{\phi_{m,\beta}\left\{
\sigma_{L^{+}}=\epsilon^{+}\right\}  \phi_{m,\beta}\left\{  \sigma_{L^{-}%
}=\epsilon^{-}\right\}  }\phi_{m,\beta}\left\{  \sigma_{L^{+}}=\epsilon
^{+}\right\}  \phi_{m,\beta}\left\{  \sigma_{L^{-}}=\epsilon^{-}\right\} \\
&  \leq e^{\psi\left(  c\right)  }\phi_{m,\beta}\left\{  \sigma_{L^{+}%
}=\epsilon^{+}\right\}  \phi_{m,\beta}\left\{  \sigma_{L^{-}}=\epsilon
^{-}\right\} \nonumber
\end{align}
and, by the symmetry under the reflection w.r.t. the horizontal axis,
\begin{align}
\phi_{n,\beta}\left\{  \sigma_{L^{+}}=\sigma_{L^{-}}\right\}   &
=\sum_{\epsilon\in\Omega_{L}}\frac{\phi_{n,\beta}\left\{  \sigma_{L^{+}%
}=\sigma_{L^{-}}=\epsilon\right\}  }{\phi_{n,\beta}\left\{  \sigma_{L^{+}%
}=\epsilon\right\}  \phi_{n,\beta}\left\{  \sigma_{L^{-}}=\epsilon\right\}
}\phi_{n,\beta}\left\{  \sigma_{L^{+}}=\epsilon\right\}  \phi_{n,\beta
}\left\{  \sigma_{L^{-}}=\epsilon\right\} \\
&  \geq e^{-\psi\left(  c\right)  }\sum_{\epsilon\in\Omega_{L}}\phi_{n,\beta
}^{2}\left\{  \sigma_{L^{+}}=\epsilon\right\}  \ .\nonumber
\end{align}
Therefore, by (\ref{L3}), (\ref{b1}) is bounded by
\begin{align}
&  e^{\frac{\psi\left(  c\right)  }{2}\left(  e^{-cm}+4\right)  }\frac
{\psi\left(  c\right)  }{2}e^{-cm}\sum_{\epsilon^{+},\epsilon^{-}\in\Omega
_{L}}b\left(  \epsilon^{+}\right)  b\left(  \epsilon^{-}\right)  \frac
{\phi_{m,\beta}\left\{  \sigma_{L^{+}}=\epsilon^{+}\right\}  \phi_{m,\beta
}\left\{  \sigma_{L^{-}}=\epsilon^{-}\right\}  }{\sqrt{\sum_{\epsilon^{+}%
\in\Omega_{L}}\phi_{m,\beta}^{2}\left\{  \sigma_{L^{+}}=\epsilon^{+}\right\}
}\sqrt{\sum_{\epsilon^{-}\in\Omega_{L}}\phi_{m,\beta}^{2}\left\{
\sigma_{L^{-}}=\epsilon^{-}\right\}  }}\times\\
&  \times\frac{\sum_{\epsilon\in\Omega_{L}}\phi_{m,\beta}^{2}\left\{
\sigma_{L^{+}}=\epsilon\right\}  }{\sum_{\epsilon\in\Omega_{L}}\phi_{n,\beta
}^{2}\left\{  \sigma_{L^{+}}=\epsilon\right\}  }\ .\nonumber
\end{align}
Proceeding as in the proof of (\ref{L2}), for any $\epsilon\in\Omega_{L},$ we
get the bound
\begin{equation}
\frac{\phi_{n,\beta}\left\{  \sigma_{L^{+}}=\epsilon\right\}  }{\phi_{m,\beta
}\left\{  \sigma_{L^{+}}=\epsilon\right\}  }\geq e^{-\psi\left(  c\right)
}\ .
\end{equation}
Hence, (\ref{b1}) is smaller than
\begin{equation}
e^{\frac{\psi\left(  c\right)  }{2}\left(  e^{-cm}+6\right)  }\frac
{\psi\left(  c\right)  }{2}e^{-cm}\left\vert \sum_{\epsilon\in\Omega_{L}%
}b\left(  \epsilon\right)  \frac{\phi_{m,\beta}\left\{  \sigma_{L^{+}%
}=\epsilon\right\}  }{\sqrt{\sum_{\epsilon^{\prime}\in\Omega_{L}}\phi
_{m,\beta}^{2}\left\{  \sigma_{L^{+}}=\epsilon^{\prime}\right\}  }}\right\vert
^{2}%
\end{equation}
and by the Schwarz inequality we get (\ref{mext}).
\end{proof}

\end{document}